%% file: arxiv.tex
\documentclass[sigplan,screen,natbib=false]{acmart}
\acmConference{LICS}{2026}{Lisbon, Portugal}
\usepackage{xspace,amsmath,amsfonts,bm,bbm,ifthen,mathtools}
\usepackage[heavycircles]{stmaryrd}
\usepackage[sort, noadjust]{cite}
\usepackage{leftidx}
\usepackage{tikz}
\usetikzlibrary{arrows, automata, shapes, calc, fadings, patterns, petri}
\tikzset{->, auto, >=stealth', font=\small}
\tikzset{state/.style={shape=circle, draw, fill=white, initial text=,
    inner sep=.5mm, minimum size=1.5mm}}
\tikzset{accepting/.style=accepting by arrow}
\tikzset{state with output/.style={shape=rectangle split, rectangle
    split parts=2, draw, fill=white,
    initial text=, inner sep=1mm}}
\tikzset{place/.style={shape=circle, draw, minimum size=4mm}}
\tikzset{transition/.style={fill,minimum width=5mm,minimum height=1mm, inner sep=.1mm}}

1

\renewcommand\thefootnote{\@arabic\c@footnote}

\renewcommand*\epsilon{\varepsilon}
\renewcommand*\phi{\varphi}

\newcommand{\sq}{\square}
\newcommand*\isq{{\textup{\textsc{i}}}\mkern-1mu \sq}
\newcommand*\csq{{\textup{c}} \sq}
\newcommand*\op{\textup{\textsf{op}}}
\newcommand*\cat[1]{\text{\textup{\textsf{#1}}}}
\newcommand*\Set{\cat{Set}}
\newcommand*\Rel{\cat{Rel}}
\newcommand*\laxto{\mathrel{\tilde\rightarrow}}
\newcommand*\evord{\dashrightarrow}
\newcommand*\ilo[3]{\leftidx{_{#1}}{#2}{_{#3}}}
\newcommand*\ibullet{\vcenter{\hbox{\tiny $\bullet$}}}
\newcommand*\pibullet{\phantom{\ibullet}}
\newcommand*\loset[1]{\left[\begin{smallmatrix}#1\end{smallmatrix}\right]}
\newcommand*\ie{\textit{i.e.},\xspace}
\newcommand*\starter[2]{\leftidx{_{#2\!}}{{\uparrow}#1}{}}
\newcommand*\terminator[2]{{#1}{\downarrow}_{#2}}
\newcommand*\id{\textup{\textsf{id}}}
\newcommand*\iPoms{\cat{iPoms}}
\newcommand*\down{\mathord{\downarrow}}
\newcommand*\dotsq{\ensuremath{\text{$\sq$\llap{$\cdot$\hspace*{.4ex}}}}}
\newcommand*\dotisq{\ensuremath{\text{$\isq$\llap{$\cdot$\hspace*{.4ex}}}}}
\newcommand*\ev{\textup{\textsf{ev}}}
\newcommand*\iev{\textup{\textsf{iev}}}
\newcommand*\mcal[1]{\mathcal{#1}}
\newcommand*\src{\textsf{src}}
\newcommand*\tgt{\textsf{tgt}}
\newcommand*\arrO[1]{\mathrel{\nearrow^{#1}}}
\newcommand*\arrI[1]{\mathrel{\searrow_{#1}}}
\newcommand*\ST{\textup{\textsf{ST}}}
\newcommand*\smalloset[1]{[\begin{smallmatrix}#1\end{smallmatrix}]}
\newcommand*{\bst}{\textup{bst}}
\newcommand*{\btt}{\textup{btt}}
\newcommand*\iST{\cat{iST}}
\newcommand*\Id{\textup{\textsf{Id}}}

\begin{document}

\title{Variants of Higher-Dimensional Automata}

\author{Hugo Bazille}
\affiliation{%
  \institution{EPITA Research Laboratory (LRE)}
  \city{Rennes}
  \country{France}}
\author{Jérémy Dubut}
\affiliation{%
  \institution{LIX, Ecole polytechnique}
  \city{Palaiseau}
  \country{France}}
\author{Uli Fahrenberg}
\affiliation{%
  \institution{Formal Methods Laboratory \& Paris-Saclay University}
  \city{Gif-sur-Yvette}
  \country{France}}
\author{Krzysztof Ziemia\'nski}
\affiliation{%
  \institution{University of Warsaw}
  \city{Warsaw}
  \country{Poland}}

\begin{abstract}
The theory of higher-dimensional automata (HDAs) has seen rapid progress in recent years, and first applications, notably to Petri net analysis, are starting to show.  It has, however, emerged that HDAs themselves often are too strict a formalism to use and reason about.  In order to solve specific problems, weaker variants of HDAs have been introduced, such as HDAs with interfaces, partial HDAs, ST-automata or even relational HDAs.

In this paper we collect definitions of these and a few other variants into a coherent whole and explore their properties and translations between them.  We show that with regard to languages, the spectrum of variants collapses into two classes, languages closed under subsumption and those that are not.  We also show that partial HDAs admit a Kleene theorem and that, contrary to HDAs, they are determinizable.
\end{abstract}


\maketitle

\section{Introduction}
\label{s:intro}

The theory of higher-dimensional automata (HDAs) has seen rapid progress in recent years \cite{%
  DBLP:journals/mscs/FahrenbergJSZ21,
  DBLP:journals/lmcs/FahrenbergJSZ24,
  DBLP:journals/fuin/FahrenbergZ24,
  journals/tcs/AmraneBZ25,
  DBLP:conf/fscd/PassemardAF25},
and first applications, notably to Petri net analysis, are starting to show \cite{%
  DBLP:conf/apn/AmraneBFHS25,
  conf/msr/AmraneBFS25}.
In these and other works, a common theme is that HDAs themselves,
as introduced in \cite{%
  Glabbeek91-hda,
  DBLP:journals/tcs/Glabbeek06}
and taken back up in \cite{DBLP:journals/mscs/FahrenbergJSZ21},
often are too strict a formalism to use and reason about:

In \cite{DBLP:journals/lmcs/FahrenbergJSZ24}, the authors introduce HDAs with \emph{interfaces} (iHDAs)
in order to be able to properly glue HDAs (necessary for a Kleene theorem).
These are taken up again in \cite{DBLP:journals/fuin/FahrenbergZ24}
so as to have a nicer Nerode-type construction turning regular languages into iHDAs.

In \cite{DBLP:conf/apn/AmraneBFHS25}, the authors revise and extend the translation of Petri nets into HDAs from \cite{DBLP:journals/tcs/Glabbeek06}.
When Petri nets are extended with inhibitor arcs, the translation gives \emph{partial} HDAs (pHDAs) where some faces may be missing.
These had been introduced earlier \cite{DBLP:conf/fossacs/Dubut19} in order to properly represent unfoldings and trees.

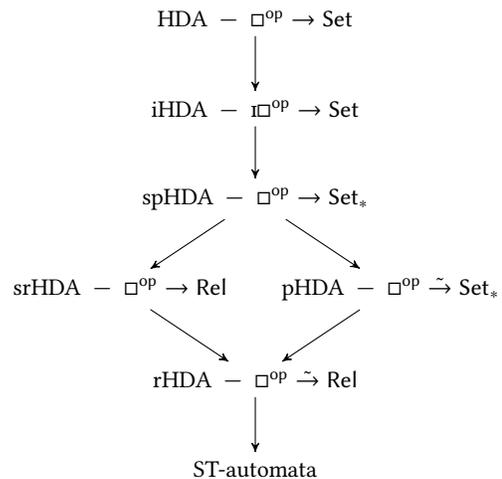
\begin{figure}[bp]
  \centering
  \newcommand*\sepa{\;---\;\ }
  \begin{tikzpicture}[x=1.8cm, y=1.2cm]
    \node (HDA) at (0,0) {HDA \sepa $\sq^\op\to \Set$};
    \node (iHDA) at (0,-1) {iHDA \sepa $\isq^\op\to \Set$};
    \node (spHDA) at (0,-2) {spHDA \sepa $\sq^\op\to \Set_*$};
    \node (srHDA) at (-1,-3) {\vphantom{p}srHDA \sepa $\sq^\op\to \Rel$};
    \node (pHDA) at (1,-3) {pHDA \sepa $\sq^\op\laxto \Set_*$};
    \node (rHDA) at (0,-4) {rHDA \sepa $\sq^\op\laxto \Rel$};
    \node (STA) at (0,-5) {ST-automata}; 
    \path (HDA) edge (iHDA);
    \path (iHDA) edge (spHDA);
    \path (spHDA) edge (srHDA) edge (pHDA);
    \path (srHDA) edge (rHDA);
    \path (pHDA) edge (rHDA);
    \path (rHDA) edge (STA);
  \end{tikzpicture}
  \caption{%
    Variants of HDAs together with their presheaf formulations (lax presheaves denoted $\laxto$).}
  \label{fig:hdaschema}
\end{figure}

The theme of missing faces is also seen in \cite{DBLP:conf/fscd/PassemardAF25}
which introduces HDAs on infinite words and notices that the usual equivalence between Muller and Büchi acceptance
does not hold for HDAs but conjectures that this problem disappears when passing to pHDAs.
Also note that \cite{DBLP:conf/calco/FahrenbergL15} introduces a different notion of pHDAs,
incompatible with the one of \cite{DBLP:conf/fossacs/Dubut19} and superceded by the latter.

Another notion which has raised its head multiple times is that of \emph{ST-automata}.
These are a variant of finite automata, on an infinite alphabet of so-called starters and terminators,
and have been used in \cite{%
  DBLP:journals/lmcs/FahrenbergJSZ24,
  DBLP:journals/fuin/FahrenbergZ24,
  DBLP:conf/fscd/PassemardAF25,
  conf/ramics/AmraneBCFZ24}
and other works to give operational semantics to HDAs and facilitate proofs.

Finally, recent unpublished work \cite{conf/fossacs/ChamounM26}
shows that \emph{relational} presheaves, specifically relational HDA,
may have useful applications in modeling concurrent systems.

In this paper we collect definitions of all these and a few other variants into a coherent whole,
explore their properties and translations between them,
and prove a result on determinizability and a Kleene theorem.
Central to this treatment are formulations as \emph{presheaves},
over different categories and both lax and strict,
which allow for a systematic analysis.
Figure~\ref{fig:hdaschema} shows the different variants of HDAs we consider here,
together with their presheaf formulation and their relationship;
see the following sections for details.

Specifically, we show that with respect to languages,
the spectrum of HDA variants collapses into two classes:
HDAs and iHDAs whose languages are closed under subsumption,
and all others which define a richer class of languages.
We~also show that, contrary to HDAs and iHDAs \cite{DBLP:journals/fuin/FahrenbergZ24}, partial HDAs are determinizable.
Finally, and similarly to HDAs and iHDAs \cite{DBLP:journals/lmcs/FahrenbergJSZ24},
pHDAs admit a Kleene theorem, but with a much simpler proof than the one for HDAs in \cite{DBLP:journals/lmcs/FahrenbergJSZ24}.

\section{Pomsets with Interfaces}

Let $\Sigma$ be a finite alphabet.
A \emph{concurrency list}, or \emph{conclist},
$U=(U, {\evord_U}, \lambda_U)$
consists of a finite set $U$,
a strict total order ${\evord_U}\subseteq U\times U$ (the event order),\footnote{%
  A strict \emph{partial} order is a relation which is irreflexive and transitive;
  a strict \emph{total} order is a relation which is irreflexive, transitive, and total;
  an \emph{acyclic} relation is one whose transitive closure is a strict partial order.}
and a labelling $\lambda_U: U\to \Sigma$.
The set of conclists over $\Sigma$ is denoted $\sq=\sq(\Sigma)$.
An isomorphism of conclists is a label-preserving order bijection.

A \emph{conclist with interfaces}, or \emph{iconclist},
is a tuple $(S, U, T)$, often written $\ilo{S}{U}{T}$,
consisting of a conclist $U$ and two subsets $S, T\subseteq U$.
These are called \emph{interfaces}; $S$ is the starting interface and $T$ the terminating interface.
The set of iconclists over $\Sigma$ is denoted $\isq=\isq(\Sigma)$.
An isomorphism of iconclists is an order bijection which preserves labels and interfaces.

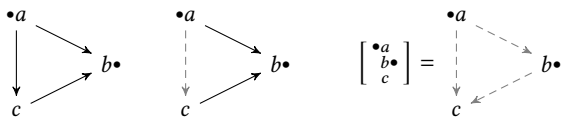
\begin{figure}[bp]
  \centering
  \begin{tikzpicture}[x=.9cm, y=.9cm]
    \begin{scope}[shift={(0,0)}]
      \node (a) at (0.4,0.7) {$\ibullet a$};
      \node (c) at (0.4,-0.7) {$c$};
      \node (b) at (1.8,0) {$b\ibullet$};
      \path (a) edge (b) (c) edge (b) (a) edge (c);
    \end{scope}
    \begin{scope}[shift={(2.5,0)}]
      \node (a) at (0.4,0.7) {$\ibullet a$};
      \node (c) at (0.4,-0.7) {$c$};
      \node (b) at (1.8,0) {$b\ibullet$};
      \path (a) edge (b) (c) edge (b);
      \path[densely dashed, gray]  (a) edge (c);
    \end{scope}
    \begin{scope}[shift={(6.5,0)}]
      \node (q) at (-0.5,0) {$ \loset {\ibullet a \pibullet \\ \pibullet b \ibullet\\ \pibullet c \pibullet}=$};
      \node (a) at (0.4,0.7) {$\ibullet a$};
      \node (c) at (0.4,-0.7) {$c$};
      \node (b) at (1.8,0) {$b\ibullet$};
      \path[densely dashed, gray] (a) edge (b) (b) edge (c) (a) edge (c);
    \end{scope}
  \end{tikzpicture}
  \caption{%
    Two ipomsets and an iconclist.
    Full arrows indicate precedence order,
    dashed arrows event order, and
    bullets, interfaces.}
  \label{fig:exmipoms}
\end{figure}

A \emph{pomset with interfaces}, or \emph{ipomset},
is a tuple $(P, {<_P}, {\evord_P},$ $S_P, T_P, \lambda_P)$
consisting of a finite set $P$,
relations ${<_P}, {\evord_P}\subseteq P\times P$
such that $<_P$ is a strict partial order and $\evord_P$ is acyclic,
a labelling function $\lambda_P : P \to \Sigma$,
and subsets $S_P, T_P \subseteq P$, called \emph{source} and \emph{target interfaces},
such that elements of $S_P$ are $<_P$-minimal and those of $T_P$ are $<_P$-maximal.
We further require that for every $x, y\in P$, exactly one of
$x=y$, $x<y$, $y<x$, $x\evord y$, and $y\evord x$ holds.
When no ambiguity may arise, we drop the subscripts to ease notation.

The intuition is that $<$ denotes \emph{precedence} of events,
whereas $\evord$ orders events $x$, $y$ which are concurrent, \ie for which neither $x<y$ nor $y<x$.
Figure~\ref{fig:exmipoms} shows some examples.
We denote interfaces using bullets and often omit the event order which by default goes downward.

An \emph{isomorphism} of ipomsets $P$ and $Q$
is a bijection $f: P\to Q$ for which
$f(S_P)=S_Q$, $f(T_P)=T_Q$, $\lambda_Q\circ f=\lambda_P$,
$f(x)<_Q f(y)$ iff $x<_P y$, and
$f(x)\evord_Q f(y)$ iff $x\evord_P y$.
$P$ and $Q$ being isomorphic is denoted $P\simeq Q$.
Isomorphisms between ipomsets are unique, so we may switch freely between ipomsets and their equivalence classes.

We recall the definition of the gluing operation of ipomsets.
Let $P$ and $Q$ be two ipomsets with $T_P\simeq S_Q$.
The \emph{gluing} of $P$ and $Q$ is defined
as $P*Q=(R, {<}, {\evord}, S, T, \lambda)$ where:
\begin{enumerate}
\item
  $R=\mathrm{colim}(P\hookleftarrow T_P\cong S_Q\hookrightarrow Q)$,
  the quotient of $P\sqcup Q$ by the identification  $T_P\cong S_Q$, and $i:P\to R$ and $j:Q\to R$ are the canonical maps;
 \item ${<}=
  \{(i(x), i(y))\mid x<_P y\}
  \cup \{(j(x), j(y))\mid x<_Q y\}
  \cup \{(i(x), j(y))\mid x\in P\setminus T_P, y\in Q\setminus S_Q\}$;
\item ${\evord}=\{(i(x), i(y)\mid x\evord_P y\}\cup \{(j(x), j(y)\mid x\evord_Q y\}$;
\item $S=i(S_P)$, $T=j(T_Q)$, $\lambda(i(x))=\lambda_P(x)$, and $\lambda(j(x))=\lambda_Q(x)$.
\end{enumerate}
Hence the precedence order in $P*Q$ puts events in $P$ before events in $Q$,
except for those which are in the terminating interface of $P$
(which is also the starting interface of $Q$)
for which no new order is introduced.
It is clear that gluing respects isomorphisms.

An ipomset $P$ is \emph{discrete} if $<_P$ is empty, hence $\evord_P$ is total.
Discrete ipomsets are precisely iconclists.
A \emph{starter} is an iconclist $\starter{U}{A}=\ilo{U\setminus A}{U}{U}$,
and a \emph{terminator} is $\terminator{U}{A}=\ilo{U}{U}{U\setminus A}$;
these are \emph{proper} if $A\ne \emptyset$.
The \emph{identity iconclists} are $\id_U=\ilo UUU$.
A \emph{step decomposition} of an ipomset $P$ is a presentation
$P=P_1*\dotsc* P_n$
as a gluing of starters and terminators.
The step decomposition is
\emph{sparse} if it is an alternating sequence of proper starters and terminators, or it is a single identity.

\begin{figure}[tp]
  \begin{tikzpicture}[y=.9cm]
  \begin{scope}
    \node (a) at (0,0) {$\ibullet a$};
    \node (b) at (0,-1) {$\pibullet b$};
    \path (a) edge[densely dashed, gray] (b);
    \node (c) at (1.5,0) {$c \pibullet $};
    \node (d) at (1.5,-1) {$d \ibullet$};
    \path (c) edge[densely dashed, gray] (d);
    \path (c) edge[densely dashed, gray] (b);
    \path (a) edge (c);
    \path (a) edge (d);
    \path (b) edge (d);
  \end{scope}
  \begin{scope}[xshift=5cm]
    \node () at (0,-0.5) {$\loset {\ibullet a \ibullet \\ \pibullet b \ibullet}*
                           \loset {\ibullet a \pibullet \\ \ibullet b \ibullet}*
                           \loset {\pibullet c \ibullet \\ \ibullet b \ibullet}*
                           \loset {\ibullet c \ibullet \\ \ibullet b \pibullet}*
                           \loset {\ibullet c \ibullet \\ \pibullet d \ibullet}*
                           \loset {\ibullet c \pibullet \\ \ibullet d \ibullet}$};
  \end{scope}
  \end{tikzpicture}
  \caption{Ipomset (left) and its sparse step decomposition.}
  \label{fig:ipomsetdecomposition}
\end{figure}
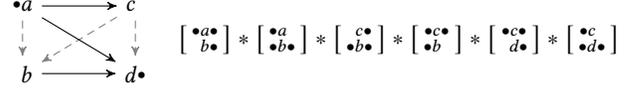

An ipomset $P$ admits a step decomposition if and only if $<_P$ 
is an \emph{interval order}~\cite{book/Fishburn85},
and in that case it admits a \emph{unique} sparse step decomposition \cite[Prop.~3.5]{DBLP:journals/fuin/FahrenbergZ24}.
All ipomsets in this paper will be interval, and the set of them is denoted $\iPoms$.
Figure~\ref{fig:ipomsetdecomposition} presents an example.

An ipomset may be \emph{refined} by adding precedence order (or equivalently removing concurrency).
The opposite operation (adding concurrency) is called \emph{subsumption}.
More formally, a \emph{subsumption} of an ipomset $P$ by $Q$
is a bijection $f: P\to Q$ for which
$f(S_P)=S_Q$, $f(T_P)=T_Q$, $\lambda_Q\circ f=\lambda_P$;
$f(x)<_Q f(y)$ implies $x<_P y$; and
$x\evord_P y$ implies $f(x)\evord_Q f(y)$.
We denote this subsumption $P \sqsubseteq Q$.
The subsumption closure of a set of ipomsets $L$ is $L\down = \{P | \exists Q \in L, P \sqsubseteq Q\}$.
A set of ipomsets $L$ is \emph{subsumption closed} if $L = L\down$.

\section{Variants of HDAs}

\subsection{The Category $\sq$}

Most variants of HDAs to be defined below are different kinds of presheaves on a common base category $\sq$,
so we take some care to introduce this one first.

A \emph{conclist inclusion} is a function $d_{A, B}: U\to V$,
for $U$ and $V$ conclists and $A, B\subseteq V$ such that
$A\cap B=\emptyset$ and $V\setminus U=A\cup B$, which is the identity on $U$. 
That is, $U$ is extended to $V$ by disjoint sets $A$ and $B$.
The composition of the conclist inclusions $d_{A, B}: U\to V$ and $d_{C, D}: V\to W$ is $d_{A\cup C, B\cup D}$.
Let $\dotsq$ be the category with objects $\sq$ and 
morphisms conclist inclusions $U\to V$.

Now let $\sq$ be the category generated by $\dotsq$ together with conclist isomorphisms on each object.
(Note the overloading of $\sq$ which will not pose problems below.)
Morphisms in $\sq(U, V)$ may be seen as injections $U\hookrightarrow V$ which preserve order and labelling
(called lo-maps in \cite{DBLP:journals/lmcs/FahrenbergJSZ24})
together with a presentation $V\setminus U=A\cup B$ (with $A\cap B=\emptyset$).
The intuition is that the events of $U$ are included in $V$,
while $A$ and $B$ are events in $V$ which have not yet started in $U$, respectively have terminated in $U$.

The category $\sq$ does not have any non-trivial automorphisms, indeed $\sq(U, U)=\{\id_U=d_{\emptyset,\emptyset}\}$ for every object $U$.
Hence $\sq$ is uniquely isomorphic to its quotient under conclist isomorphisms.
We will also denote that quotient by $\sq$ and switch freely between the two versions.
The objects of the quotient are conclists without event identity: only the labelling and the order of events matters.

\subsection{HDAs and HDAs with Interfaces}

A \emph{precubical set} is a presheaf on $\sq$, that is, a functor $X: \sq^\op\to \Set$.
For $x\in X[U]$ we write $\ev(x)=U$, defining a function $\ev: X\to \sq$.
We write $\delta_{A, B}=X[d_{A, B}]$ and call these \emph{face maps};
omitting $U$ from the notation will not cause trouble.

In elementary terms, that means that a precubical set is a set $X$ (of \emph{cells})
together with a function $\ev: X\to \sq$
and face maps $\delta_{A, B}$ which satisfy the \emph{precubical identity}
\begin{equation}
  \label{eq:pcid}
  \delta_{C, D} \delta_{A, B}=\delta_{A\cup C, B\cup D}.
\end{equation}
In a face $y=\delta_{A, B}(x)$, the events in $A\subseteq \ev(x)$ have been \emph{unstarted}
and the events in $B\subseteq \ev(x)$ have been \emph{terminated}.
We write $\delta_A^0=\delta_{A, \emptyset}$ and $\delta_B^1=\delta_{\emptyset, B}$ for lower and upper faces.

\begin{definition}
  \quad
  A \emph{higher-dimensional automaton} (HDA) $(X, \bot, \top)$ consists of a precubical set $X$
  together with subsets $\bot, \top$ of \emph{initial} and \emph{accepting} cells.
\end{definition}

\begin{figure}[bp]
  \begin{tikzpicture}[y=.9cm]
    \node (ss1) at (-3,-2) {$\ibullet a\pibullet$};
    \node (ss2) at (-3,-3) {$\ibullet b\ibullet$};
    \path (ss1) edge[densely dashed, gray] (ss2);
    \node (t1) at (0,-1) {$\ibullet a\pibullet$};
    \node (t2) at (0,-2) {$\pibullet c\ibullet$};
    \node[right] at (t2) {$\;\;A$};
    \node (t3) at (0,-3) {$\ibullet b\ibullet$};
    \node (t4) at (0,-4) {$\ibullet e\pibullet$};
    \node[right] at (t4) {$\;\;B$};
    \path (t1) edge[densely dashed, gray] (t2);
    \path (t2) edge[densely dashed, gray] (t3);
    \path (t3) edge[densely dashed, gray] (t4);
    \path (ss1) edge (t1);
    \path (ss2) edge (t3);
  \end{tikzpicture}
  \caption{An iconclist map.
    Annotations $A$ and $B$
    indicate events that have not yet started resp.~terminated.
    Being in the target interface, $c\ibullet$ must be marked $A$ rather than $B$.}
  \label{fig:iconclistMap}
\end{figure}
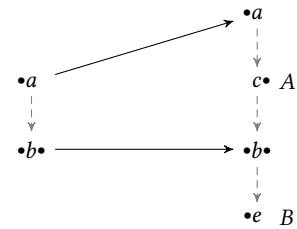

We now proceed to HDAs with interfaces which are essentially presheaves on a category of iconclists.
An \emph{iconclist inclusion} $d_{A, B}: (S_U, U, T_U)\to (S_V, V, T_V)$,
for $(S_U, U, T_U)$ and $(S_V, V, T_V)$ iconclists,
is a conclist inclusion $U\to V$ such that $A\cap S_V=B\cap T_V=\emptyset$, $S_U=S_V\setminus B$, and $T_U=T_V\setminus A$.
The intuition is that the interfaces mark events whose status may not be changed:
the $A$-events in $V$, not yet started in $U$, cannot be part of the starting interfaces $S_V$ neither $S_U$,
and the $U$-terminated $B$-events in $V$ cannot be part of the terminating interfaces.
See Fig.~\ref{fig:iconclistMap} for an example.

Let $\dotisq$ be the category with objects $\isq$,
morphisms iconclist inclusions,
and composition defined as for $\dotsq$,
and let $\isq$ be the category generated by $\dotisq$ together with iconclist isomorphisms on each object.
Again, $\isq$ is uniquely isomorphic to its quotient under iconclist isomorphisms,
and we will switch freely between $\isq$ and its quotient.

A \emph{precubical set with interfaces} is a presheaf on $\isq$, \ie a functor $X: \isq^\op\to \Set$.
We write $\iev(x)=\ilo{S}{U}{T}$ for $x\in X[\ilo{S}{U}{T}]$ and $\delta_{A, B}=X[\ilo{A}{U}{B}]$.

\begin{definition}
  A \emph{higher-dimensional automaton with interfaces} (iHDA) $(X, \bot, \top)$ consists of a precubical set with interfaces $X$
  together with subsets $\bot, \top$ of \emph{initial} and \emph{accepting} cells
  such that $S=U$ for all $x\in \bot$ and $T=U$ for all $x\in \top$ with $\iev(x)=\ilo{S}{U}{T}$.
\end{definition}

We will later see that iHDAs are strict partial HDAs in which the partiality is event-based,
in the sense that whether or not $\delta_{A, B}(x)$ is defined for a cell $x$ depends only on $\iev(x)$.

\subsection{Partial and Relational HDAs}

Let $\Set_*$ be the category of sets and partial functions (\ie pointed sets) and $\Rel$ the category of relations.
A \emph{partial precubical set} is a lax $\Set_*$-valued presheaf on $\sq$
(a lax functor $\sq^\op\laxto \Set_*$),
and a \emph{relational precubical set} is a lax $\Rel$-valued presheaf on $\sq$
(a lax functor $\sq^\op\laxto \Rel$).
As above, we denote lax functors using $\laxto$; see \cite{journals/tac/Niefield10} for definitions and context.

In elementary terms, a relational precubical set thus consists of a set $X$,
a function $\ev: X\to \sq$,
and face relations $\delta_{A, B}\subseteq X[U]\times X[U\setminus(A\cup B)]$ which satisfy the \emph{lax precubical identity}
\begin{equation}
  \label{eq:lpcid}
  \delta_{C, D} \delta_{A, B}\subseteq\delta_{A\cup C, B\cup D}
\end{equation}
(we write relational composition from right to left).
In a partial precubical set, the face relations are partial functions;
we will generally prefer to write relations and partial functions as multi-valued functions below.

We will also have occasion to use strict variants of the above presheaves,
\ie non-lax functors $\sq^\op\to \Set_*$ and $\sq^\op\to \Rel$;
in that case the identity \eqref{eq:lpcid} is replaced by \eqref{eq:pcid}.

\renewcommand*\thedefinition{\arabic{section}.\arabic{definition}(s)p/r}
\begin{definition}
  A (\emph{strict}) \emph{partial} / \emph{relational} \emph{higher-dim\-en\-sio\-nal automaton} ((s)p/rHDA) $(X, \bot, \top)$
  consists of a (strict) partial / relational precubical set $X$
  together with subsets $\bot, \top$ of \emph{initial} and \emph{accepting} cells.
\end{definition}
\renewcommand*\thedefinition{\arabic{section}.\arabic{definition}}

Note that the above definition introduces four different notions.
We detail these and
give an intuition about their differences with HDAs.
The category $\Set_*$
formalizes the fact that some face maps may be ``missing''.
It can be because the face does not exist, as depicted in Fig.~\ref{fig:diff_sphda_phda} below
where the two-dimensional square $x$ has only one lower face, as $\delta^0_{a}(x)$ is not defined.
It can also be because while the cell exists, the face map does not.
For example in Fig.~\ref{fig:phda_non_geometric}, $\delta^1_{\{a, b\}}(x)$, $\delta^1_{b}(x)$ and $\delta^1_{a}(x)$ are defined, but $\delta^1_{a}\delta^1_b(x)$ is not: if event $b$ terminates before $a$ does, then $a$ cannot terminate.
However, they can terminate simultaneously.

\input{figures/phda_non_geometric}

In addition of being partial, face maps respect \eqref{eq:lpcid}.
Here the notation $\delta_{C, D} \delta_{A, B}\subseteq\delta_{A\cup C, B\cup D}$ means that
if $\delta_{A, B}$ and $\delta_{C, D}$ are defined,
then also $\delta_{A\cup C, B\cup D}$ is defined and equal to the composition $\delta_{C, D} \delta_{A, B}$;
but $\delta_{A\cup C, B\cup D}$
may be defined without one or both of the maps on the left-hand side being defined.

In spHDAs, face maps respect the stricter relation \eqref{eq:pcid}.
As a consequence, if a face is missing, then so are its own lower and upper faces.

In rHDAs, as there are face relations instead of face maps,
a cell may have zero, one, or several lower/upper faces for a given set of events.
This is illustrated in Fig.~\ref{fig:diff_srhda_rhda},
where the two-dimensional cell $x$ has two lower faces when $a$ and $b$ are unstarted,
with both cells depicted in the lower left corner.

In srHDAs, face relations again respect the stricter (\ref{eq:pcid}).
This is illustrated in Fig.~\ref{fig:diff_sphda_srhda},
where again $x$ has two lower faces when $a$ and $b$ are unstarted, but now the lower and left edges are connected to both corners.

\subsection{ST-Automata}

We now define ST-automata and some (new) generalizations which will be useful later.

\begin{definition}
  A \emph{P-automaton} is a tuple
  $\mcal{A}=(Q, E, s, t, \lambda, \mu,$ $\bot, \top)$,
  where $Q$ is a finite set of \emph{states},
  $E$ is a finite set of \emph{transitions},
  $s, t: E\to Q$ are the \emph{source} and \emph{target} functions,
  $\mu: Q\to \sq$ and $\lambda: E\to \iPoms$ are \emph{labellings},
  and $\bot, \top\subseteq Q$ are initial and accepting states.
  We require that $\mu(s(e))=S_{\lambda(e)}$ and $\mu(t(e))=T_{\lambda(e)}$ for every transition $e\in E$.
\end{definition}

Also P-automata are presheaves, but that observation is of little interest to us here and we will not use it.

\begin{definition}
  A P-automaton is 
  \begin{itemize}
  \item an \emph{ST-automaton} if $\lambda(e)$ is a starter or a terminator;
  \item a \emph{gST-automaton} if $\lambda(e)$ is discrete;
  \item \emph{proper} if $s(e)\not\in \top$ and $t(e)\not\in \bot$;
  \item \emph{without silent transitions} if $\lambda(e)$ is not an identity; 
  \end{itemize}
  for all $e\in E$.
\end{definition}

\subsection{Paths and Languages}

A \emph{path} in a P-automaton $\mcal{A}$ is defined as usual for automata: a sequence 
$\alpha=(q_0, e_1, q_1, \dotsc, e_n, q_n)$
such that $q_i\in Q$, $e_i\in E$ and $s(e_i)=q_{i-1}$, $t(e_i)=q_i$.
The path $\alpha$ is \emph{accepting} if $q_0=\src(\alpha)\in \bot$ and $q_n=\tgt(\alpha)\in\top$.
The ipomset recognized by $\alpha$ is $\ev(\alpha)=\lambda(e_1)*\dotsm*\lambda(e_n)$.

\begin{definition}
  The \emph{language} of a P-automaton $\mcal{A}$ is
  $L(\mcal{A})=\{\ev(\alpha)\mid \alpha \text{ accepting path in } \mcal{A}\}$.
\end{definition}

\begin{lemma}
  \label{le:nosilent}
  For every P-automaton there exists a P-automaton without silent transitions that recognizes the same language.
\end{lemma}

\begin{proof}
  The standard closure construction for finite automata applies.
\end{proof}

Using the translations defined in Sect.~\ref{se:rel}, the above definition introduces languages for all models defined in this paper.
For completeness, we also define paths and languages for rHDAs.
Let $(X, \bot, \top)$ be an rHDA.
An \emph{upstep} in $X$ is $x\arrO{A} y$ for $x, y\in X$ with $x\in \delta_A^0(y)$.
A \emph{downstep} is $y\arrI{B} z$ for $y, z\in X$ with $z\in \delta_B^1(y)$.

A \emph{path} in $X$ is a sequence
$\alpha=(x_0\arrO{A_1}e_1\arrI{B_1}x_1\arrO{A_2}\dotsm \arrI{B_n} x_n)$
of up- and downsteps;
the sequence does not need to be alternating and may start with a downstep and terminate with an upstep.
If the sequence \emph{is} alternating, \ie~every upstep is followed by a downstep and vice versa,
then the path is called \emph{sparse}.

Path $\alpha$ is \emph{accepting} if $x_0\in \bot$ and $x_n\in \top$.
Paths $\alpha$ and $\beta$ may be \emph{concatenated}
if the last cell of $\alpha$ is equal to the first cell of $\beta$,
and their concatenation is written $\alpha*\beta$.
Thus, paths themselves are simply concatenations of up-and downsteps.

The ipomset recognized by an upstep $x\arrO{A} y$ is the starter $\ev(x\arrO{A} y)=\starter{\ev(y)}{A}$,
the one recognized by a downstep is $\ev(y\arrI{B} z)=\terminator{\ev(y)}{B}$.
The ipomset $\ev(\alpha)$ recognized by a path $\alpha$ is the concatenation of the ones of the steps,
and the \emph{language} of $X$ is $L(X)=\{\ev(\alpha)\mid \alpha \text{ accepting path in } X\}$.

\begin{lemma}
  \label{le:sparse_existence}
  For every path $\alpha$ in an (s)(p/r)HDA $X$ there exists a sparse path $\beta$
  such that $\src(\beta)=\src(\alpha)$, $\tgt(\beta)=\tgt(\alpha)$ and $\ev(\beta)=\ev(\alpha)$.
\end{lemma}

\begin{proof}
  Using \eqref{eq:lpcid},
  every pair of consecutive upsteps or downsteps can be merged into a single one.
\end{proof}

\begin{lemma}
  \label{le:decomp}
  Let $P,Q\in\iPoms$, $T_P=S_Q$. Assume that the sparse step decomposition of $P*Q$ is a concatenation of sparse step decompositions of $P$ and $Q$. Then every path $\alpha$ such that $\ev(\alpha)=P*Q$ can be decomposed as $\alpha=\beta*\gamma$, where $\ev(\beta)=P$ and $\ev(\gamma)=Q$.
\end{lemma}

\begin{proof}
  The lemma is an extension of \cite[Lemma~4.13]{DBLP:journals/fuin/FahrenbergZ24}, as is its proof.
\end{proof}

\section{Example}

\begin{figure}
  \centering
  \begin{tikzpicture}[x=.8cm, y=1.1cm]
    \begin{scope}
      \node[place, label=above:$p_1$, tokens=1] (1) at (0,0) {};
      \node[place, label=above:$p_2$, tokens=2] (2) at (1.5,0) {};
      \node[place, label=above:$p_3$, tokens=1] (3) at (3,0) {};
      \node[place, label=left:$p_4$] (4) at (1.5,-2) {};
      \draw[-latex, double] (2) to[bend right=1.5cm] (4);
      \node[transition, label=left:$a$] (t124) at (0.75,-1) {} edge[pre] (1) ;
      \node[transition, label=below:$b$] (t32) at (2.25,-1) {} edge[pre] (3) edge[post] (2);
    \end{scope}
  \end{tikzpicture}
  \caption{Petri net with transfer arc.}
  \label{fig:petritrans}
\end{figure}

Before we proceed, we give an example which motivates some of the definitions and constructions.
This section assumes some knowledge of Petri nets and may safely be skipped.

The paper \cite{DBLP:conf/apn/AmraneBFHS25} defines translations from various extensions of Petri nets to HDAs
and notes that the latter have to be partial for this to make sense.
As an example, Fig.~\ref{fig:petritrans} shows a Petri net with a transfer arc
from $p_2$ to $p_4$, through the transition $a$.
Whenever transition $a$ is fired, \emph{all} tokens present in $p_2$ are transferred to $p_4$.

The HDA semantics of Petri nets, introduced in \cite{DBLP:journals/tcs/Glabbeek06}
and refined in \cite{DBLP:conf/apn/AmraneBFHS25},
intends to faithfully capture both reachability and concurrency.
Here, $a$ and $b$ may be fired concurrently, but also sequentially in any order.

In the HDA semantics, firings have durations.
When a transition starts firing, it takes the tokens in its preplaces
and keeps them until it finishes firing.
So in our example, if $a$ starts firing before $b$ terminates, then two tokens are transferred to $p_4$;
but if $b$ terminates before $a$ starts,
then there will be an extra token in $p_2$ which $a$ will transfer to~$p_4$.

\input{figures/broken_phda}

Figure~\ref{fig:broken_phda} shows the partial-HDA semantics of our example net.
The cell marked $m_2$ corresponds to the marking with three tokens in $p_4$
which is reached when terminating $b$ before starting $a$;
all other paths reach the cell marked $m_1$ with two tokens in $p_4$.
The pHDA has no missing cells, but its geometry is ``broken'', in the sense that
$\delta_a^0(\delta_b^1(x))$ is undefined, but $\delta_b^1(\delta_a^0(x))$ exists.
It is not a strict pHDA.

\section{Relationship}
\label{se:rel}

We now detail the translations between the models in Fig.~\ref{fig:hdaschema}.

The inclusion of precubical sets into precubical sets with interfaces (adding empty interfaces)
is functorial and has a right adjoint, \emph{closure} which adds missing faces and in turn has a right adjoint, \emph{resolution};
see \cite{DBLP:journals/lmcs/FahrenbergJSZ24} for details.
Indeed, closure is the essential geometric morphism induced by the forgetful functor $\isq\to \sq$
\cite{DBLP:journals/apal/StruthZ26}.
Closure and resolution may be transferred to (i)HDAs and preserve languages \cite{DBLP:journals/lmcs/FahrenbergJSZ24}.

\begin{figure}
  \centering
  \begin{tikzpicture}[x=1cm, y=1cm]
    \begin{scope}
      \path[fill=black!10!white] (0,0) -- (2,0) -- (2,2) -- (0,2) -- (0,0);
      \node () at (1,1) {$x$};
      \node[state, rectangle] (00) at (0,0) {};
      \node[state, rectangle] (01) at (0,2) {};
      \node[state, rectangle] (10) at (2,0) {};
      \node[state, rectangle] (11) at (2,2) {};
      \path[draw] (00) -- (10);
      \path[draw] (01) edge node[above]  {$a$} (11);
      \path[draw] (00) -- (01);
      \path[draw] (10) edge node[right]  {$b$}  (11);
      \node[below] at (1,0) {$\bot$};
      \node[above right] at (2,2) {$\top$};
    \end{scope}
    \begin{scope}[shift={(4,0)}]
      \path[fill=black!10!white] (0,0) -- (2,0) -- (2,2) -- (0,2) -- (0,0);
      \node () at (1,1) {$x$};
      \node[] (00) at (0,0) {};
      \node[] (01) at (0,2) {};
      \node[state, rectangle] (10) at (2,0) {};
      \node[state, rectangle] (11) at (2,2) {};
      \path[draw] (00) -- (10);
      \path[draw] (01) edge node[above]  {$a$} (11);
      \path[draw] (10) edge node[right]  {$b$}  (11);
      \node[below] at (1,0) {$\bot$};
      \node[above right] at (2,2) {$\top$};
    \end{scope}
  \end{tikzpicture}
  \caption{HDA and corresponding iHDA.}
  \label{fig:hda-to-ihda}
\end{figure}

When including HDAs into iHDAs, some cells may need to be removed due to the condition
that start events may not be unstarted and accept events may not be terminated.
This removal does not change the language recognized, see Fig.~\ref{fig:hda-to-ihda} for an example.

For the translation from iHDA to spHDA, some care is required.
It involves transforming a presheaf over one category into a 
presheaf over another category, via a form of pushforward.
A convenient framework to describe this process is
when there is a 
functor $F:C\to\,D$ between small categories such that for every object $c$ of $C$ 
and every morphism $f:F(c)\to d'$ there is at most one morphism $g:c\to c'$ such that $F(g)=f$,
that is, $F$ is almost a discrete op-fibration. 
Under this condition, there is a generic way to turn a functor
$G:C\to\Set_*$ into a functor $H:D\to\Set_*$, 
in a manner reminiscent of left Kan extensions:
\begin{equation*}
  H[d]=\coprod_{c \mid F(c)=d} G[c]
\end{equation*}
for object $d$ of $D$, and for every morphism $f:d\to d'$,
\[
  H[f](x)=
  \begin{cases}
    G[g](x) & \text{if $x\in G[c]$ and $\exists\,g$ with $F[g]=f$}, \\
    * & \text{otherwise}.
  \end{cases}
\]

In the case of the forgetful functor $\isq^\op\to\sq^\op$,
this gives the construction $\text{iHDA}\to\text{spHDA}$ which translates as follows. 
For an iHDA $Y$ we define $X=I(Y)$ by
\begin{equation*}
  X[U]=\coprod_{S, T\subseteq U} Y[\ilo SUT]
\end{equation*}
with
\[
  \delta_{A,B}(y)=
  \begin{cases}
    \delta_{A,B}(y) & \text{if $A\cap S=\emptyset=B\cap T$},\\
    * & \text{otherwise}.
  \end{cases}
\]
for $y\in Y[\ilo SUT]\subseteq X[U]$ and $A,B\subseteq U$, $A\cap B=\emptyset$.
Intuitively, if a face map is ``missing'' in the $\text{iHDA}$, we define all its values in $X$ as undefined.

The inclusions $\text{spHDA}\hookrightarrow \text{srHDA}$
and $\text{pHDA}\hookrightarrow \text{rHDA}$ are induced by the forgetful functor $\Set_*\to \Rel$, and
the inclusions $\text{spHDA}\hookrightarrow \text{pHDA}$ and $\text{spHDA}\hookrightarrow \text{srHDA}$
come from the fact that every strict functor is lax.

Finally, the translation $\text{rHDA}\to \text{STA}$
was introduced for HDAs in \cite{conf/ramics/AmraneBCFZ24}; we repeat it here and extend it to rHDAs.
Let $(X, \bot, \top)$ be an rHDA,
then its \emph{operational semantics} is given by the ST-automaton
$\ST(X)=(X, E, s, t, \lambda, \ev, \bot, \top)$ with
\begin{gather*}
  \begin{aligned}
    E={} &\{(p, \starter{\ev(q)}{A}, q)\mid A\subseteq \ev(q), p\in \delta_A^0(q)\} \\
    {}\cup{} &\{(q, \terminator{\ev(q)}{B}, r)\mid B\subseteq \ev(q), r\in \delta_B^1(q)\},
  \end{aligned} \\
  s((p, U, q))=p, \quad t((p, U, q))=q, \quad
  \lambda((p, U, q))=U.
\end{gather*}
That is, the transitions of $\ST(X)$ precisely mimic the starting and terminating of events in~$X$.
Consequently, there is a one-to-one correspondence between 
paths in $X$ and paths in $\ST(X)$ that preserves 
labelling and acceptance.

\input{figures/diff_hda_ihda}

Next we give some examples to separate the classes.
Figure~\ref{fig:diff_hda_ihda} shows an iHDA $X$ with $X[\smalloset{\pibullet a \\ \ibullet b}]=\{x\}$.
(In all examples, $\bot=\top=\emptyset$.)
Given that $b$ cannot be unstarted, $X$ misses the faces $\delta_a^0(x)$, $\delta_b^0 \delta_a^0(x)$ and $\delta_b^1 \delta_a^0(x)$,
so $X$ is not an HDA.

\input{figures/diff_ihda_sphda}

\input{figures/diff_sphda_phda}

In Fig.~\ref{fig:diff_ihda_sphda} we give an spHDA $X$ with $X[\loset{a \\ b}]=\{x\}$
in which $\delta_{a, b}^0(x)=\delta_b^0(\delta_a^0(x))=\delta_a^0(\delta_b^0(x))=\emptyset$.
This is not an iHDA as the event $a$ cannot be unstarted only before $b$ has started.
Figure~\ref{fig:diff_sphda_phda} displays the same $\loset{a \\ b}$-square,
but now $\delta_a^0(x)=\emptyset$.
As $\delta_{a, b}^0(x)$ is defined but $\delta_b^0(\delta_a^0(x))$ is not, this is a pHDA but not an spHDA.

\input{figures/diff_sphda_srhda}

\input{figures/diff_srhda_rhda}

Figure~\ref{fig:diff_sphda_srhda} shows an srHDA
with $\delta_{a, b}^0(x)=\delta_b^0(\delta_a^0(x))=\delta_a^0(\delta_b^0(x))=\{u,v\}$
which is not an spHDA,
and Fig.~\ref{fig:diff_srhda_rhda} displays an rHDA
with $\delta_{a, b}^0(x)=\delta_b^0(\delta_a^0(x))=\{u,v\}$ but $\delta_a^0(x)=\emptyset$
which is neither an srHDA nor a pHDA.

\input{figures/diff_rhda_sta}

Finally, Fig.~\ref{fig:diff_rhda_sta} shows an ST-automaton
with $Q[\emptyset]=\{p\}$, $Q[a]=\{q\}$, $Q[\loset{a\\b}]=\{r\}$
and no other states and transitions $p\to q\to r$.
The composite transition $p\to r$ is missing,
violating \eqref{eq:lpcid}, so this is not the translation of an rHDA.
Strictness of the inclusions
$\text{STA}\hookrightarrow \text{gSTA}\hookrightarrow \text{PA}$ follows straight from the definition.

\section{Language Equivalence}

We show that for every P-automaton there exists a strict partial HDA that recognizes the same language.
We do so by introducing a notion of \emph{reduced} gST-automaton which essentially corresponds to a pHDA
and showing that these may in fact be translated to strict pHDA.
Using Lemma \ref{le:nosilent} we may assume that there are no silent transitions.

First, an easy lemma that P-automata may be turned into ST-automata.

\begin{lemma}
  \label{l:PtoST}
  For every P-automaton $\mcal{A}$ there exists a language equivalent ST-automaton.
\end{lemma}

\begin{proof}
	Let $e_0$ be a transition such that $\lambda(e_0)$ is neither a starter nor a terminator. Let $\lambda(e_0)=P_1*\dotsc* P_n$ be an ST-decomposition. Let $\mcal{A}'$ be a P-automaton that is obtained from $\mcal{A}$ by replacing the transition $e$ by the sequence
	
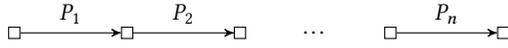
\[
	\begin{tikzpicture}
		\node[state,rectangle] (0) at (0,0) {};
		\node[state,rectangle] (1)  at (1.5,0) {};
		\node[state,rectangle] (2)  at (3,0) {};
		\node[state,rectangle] (n1) at (5,0) {};
		\node[state,rectangle] (n) at (6.5,0) {};
		\node at (4,0) {\dots};
		  \path[] (0) edge node[above]  {$P_1$} (1);
		  \path[] (1) edge node[above]  {$P_2$} (2);
		  \path[] (n1) edge node[above]  {$P_n$} (n);
	\end{tikzpicture}
	\]
	By repeating this construction we eliminate all labels that are neither starters  nor terminators.
\end{proof}

Fix a gST-automaton $\mcal{A}=(Q,E,s,t,\lambda,\mu,\bot,\top)$ for the rest of the section.
We say that $\mcal{A}$ is \emph{reduced} if is has no silent transitions and satisfies the following conditions:
\begin{itemize}
\item[(a)] $q\in\bot\implies t^{-1}(q)=\emptyset$;
\item[(b)] $q\in\top\implies s^{-1}(q)=\emptyset$;
\item[(c)] $\lambda(e)=\ilo SUT\;\land\; T=U\implies t(e)\in\top$;
\item[(d)] $\lambda(e)=\ilo SUT\;\land\; S=U\implies s(e)\in\bot$;
\item[(e)] $q\in\bot\implies |s^{-1}(q)|\leq 1$;
\item[(f)] $q\in\top\implies |t^{-1}(q)|\leq 1$;
\end{itemize}	
for all $q\in Q$ and $e\in E$.
That is, $\mcal{A}$ has no transitions incoming to initial states, neither outgoing from accepting states ((a) and (b));
the only pure starter transitions are to accepting states, and pure terminator transitions come from initial states ((c) and (d));
and initial states have at most one outgoing transition and accepting states at most one incoming transition ((e) and (f)).
We will successively transform $\mcal{A}$ in order to satisfy all these conditions,
preserving languages in the process.

\begin{lemma}
  \label{l:gSTab}
  There exists a gST-automaton $\mcal{A}'$ that satisfies (a) and (b) and $L(\mcal{A})=L(\mcal{A}')$.
\end{lemma}

\begin{proof}
  Define $\mcal{A}'=(Q',E',s',t',\lambda',\mu',\bot',\top')$ as follows:
  \begin{itemize}
  \item $Q'=Q\cup\{q_\bot\mid q\in \bot\}\cup\{q^\top\mid q\in \top\}\cup\{q_\bot^\top\mid q\in\bot\cap\top\}$,
    adding fresh initial and accepting states;
  \item $E'=E\cup\{e_\bot\mid s(e)\in \bot\}\cup \{e^\top\mid t(e)\in \top\}\cup \{e_\bot^\top\mid s(e)\in\bot,\; t(e)\in \top\}$,
    with
  \item  $s'(e_\bot)=s(e)_\bot$, $s'(e^\top)=s(e)$, $s'(e_\bot^\top)=s(e)_\bot$,
  \item $t'(e_\bot)=t(e)$, $t'(e^\top)=t(e)^\top$, $t'(e_\bot^\top)=t(e)^\top$,
  \item and otherwise $s'(e)=s(e)$ and $t'(e)=t(e)$;
  \item $\bot'=\{q_\bot\mid q\in \bot\}\cup\{q_\bot^\top\mid q\in\bot\cap\top\}$ and
    $\top'=\{q^\top\mid q\in \top\}\cup\{q_\bot^\top\mid q\in\bot\cap\top\}$; and
  \item $\lambda', \mu'$ extensions of $\lambda, \mu$ by
    $\lambda'(e_\bot)=\lambda'(e^\top)=\lambda'(e_\bot^\top)=\lambda(e)$ and
    $\mu'(q_\bot)=\mu'(q^\top)=\mu'(q_\bot^\top)=\mu(q)$.
  \end{itemize}
  
  Verifying that $\mcal{A}'$ satisfies (a) and (b) is trivial.
  Using the functional simulation $p:\mcal{A}'\to\mcal{A}$
  that sends $q,q_\bot,q^\top,q_\bot^\top$ to $q$ and $e,e_\bot,e^\top,e_\bot^\top$ to $e$,
  we have $L(\mcal{A}')\subseteq L(\mcal{A})$.

  To show the other direction, let $\alpha=(q_0, e_1, q_1,\dotsc, q_{n-1}, e_n, q_n)$ be an accepting path in $\mcal{A}$,
  then $\alpha'=((q_0)_\bot, (e_1)_\bot, q_1,\dotsc, q_{n-1},$ $(e_n)^\top, (q_n)^\top)$
  is an accepting path in $\mcal{A}'$ that recognizes the same ipomset.
  (If the length of $\alpha$ is less than $2$, we need to modify this argument.
  If $\alpha=(q, e, r)$, we take $\alpha'=(q_\bot, e_\bot^\top, r^\top)$;
  if $\alpha=(q)$ is constant, we let $\alpha'=(q_\bot^\top)$.)
\end{proof}

Now we assume that $\mcal{A}$ satisfies (a) and (b), and we modify it to assure (c) and (d).
We say that $e\in E$ is a \emph{starter transition} if $\lambda(e)$ is a starter,
its \emph{dimension} is $|\lambda(e)|$,
and it is \emph{bad} if its does not satisfy (c), \ie~$t(e)\notin \top$.
Similarly we define terminator transitions and bad terminator transitions.
Let $\bst_n(\mcal{A})$ and $\btt_n(\mcal{A})$ be the number of bad starter and bad terminator transitions,
respectively, in $\mcal{A}$.

For any starter transition $c\in E$, define a gST-automaton $\mcal{A}_{c}=(Q, E', s', t', \lambda', \mu, \bot, \top)$
by letting $E'=E\setminus \{c\}\cup \{e^*\mid e\in s^{-1}(t(c))\}$
and extending $s$, $t$, and $\lambda$ by
$s'(e^*)=s(c)$, $t'(e^*)=t(e)$, and $\lambda'(e^*)=\lambda(c)*\lambda(e)$.
That is, we remove $c$ and add transitions from the source of $c$ to all targets of transitions out of the target of $c$,
adjusting labels accordingly.
(This is similar to one the closure construction for removing silent transitions.)
Clearly $\mcal{A}_c$ still satisfies (a) and (b), as
any new transitions start or terminate at states that already had outgoing or incoming transitions, respectively.

\begin{lemma}
  \label{l:gSTcd}
  Assume that $c\in E$ is a bad starter transition of dimension $n$. Then 
  \begin{enumerate}
  \item $L(\mcal{A})=L(\mcal{A}_c)$,
  \item $\bst_n(\mcal{A}_c)=\bst_n(\mcal{A})-1$,
  \item $\bst_k(\mcal{A}_c)=\bst_k(\mcal{A})$ for $k<n$,
  \item $\btt_k(\mcal{A}_c)=\btt_k(\mcal{A})$ for all $k$,
  \end{enumerate}
\end{lemma}

\begin{proof}
  Let $\alpha=(q_0, e_1,\dotsc, e_n, q_n)$ be an accepting path in $\mcal{A}$.
  Clearly $e_n\neq c$ (since $t(c)\not\in \top$),
  and $e_i=c$ implies $e_{i+1}\neq c$ (since $\lambda(c)$ is a proper starter we have   $s(c)\neq t(c)$).
  By replacing every sequence $c, e_i$ in $\alpha$ by $e_i^*$
  we obtain an accepting path $\alpha$ in $\mcal{A}_c$ such that $\ev(\alpha)=\ev(\alpha')$.
  Thus $L(\mcal{A})\subseteq L(\mcal{A}_e)$.

  Conversely, if $\beta$ is an accepting path in $\mcal{A}_c$,
  then by replacing every step $e^*$ by a sequence $c, e$
  we obtain an accepting path in $\mcal{A}$, showing that the converse inclusion also holds.

  For the other claims, note that no $e^*$ is a terminator transition (showing 4),
  and if $e^*$ is a starter transition, then $|\lambda(e^*)|>|\lambda(c)|$ (showing (3)).
  (2) is true because we have removed $c$.
\end{proof}

There is a symmetric lemma for bad terminator transitions.

\begin{lemma}
  \label{l:gSTcd2}
  There exists a gST-automaton $\mcal{A}'$ that satisfies (a)--(d) and $L(\mcal{A})=L(\mcal{A}')$.
\end{lemma}

\begin{proof}
  Assume otherwise.
  Then there exists $k, n>0$ such that every gST-automaton $\mcal{B}$ that satisfies (a), (b) and $L(\mcal{B})=L(\mcal{A})$
  has at least $k$ bad transitions of dimension $n$.
  But then using Lemma \ref{l:gSTcd} we can remove one bad starter or terminator transition of dimension $n$
  and not introduce any new bad transitions of dimension less than or equal to $n$.
  As the dimension of transitions in finite gST-automata is bounded, this is a contradiction.
\end{proof}

Now assume that $\mcal{A}$ satisfies (a)--(d). Define a gST-automaton $\mcal{A}'=(Q',E,s',t',\lambda,\mu',\bot',\top')$ as follows:
\begin{itemize}
\item $Q'=Q\setminus (\bot\cup\top)\cup (\bot\cap \top)\cup\{e_\bot\mid e\in s^{-1}(\bot)\}\cup \{e^\top\mid e\in s^{-1}(\top)\}$,
\item $\mu'(e_\bot)=\mu(s(e))$, $\mu'(e^\top)=\mu(t(e))$,
\item $\bot'=\{e_\bot\mid e\in s^{-1}(\bot)\}\cup(\bot\cap \top)$,
\item $\top'=\{e^\bot\mid e\in t^{-1}(\top)\}\cup(\bot\cap \top)$, and
\end{itemize}
\begin{equation*}
  s'(e)=
  \begin{cases}
    s(e) & \text{if $s(e)\not\in \bot$},\\
    e_\bot & \text{if $s(e)\in \bot$},
  \end{cases}
  \quad
  t'(e)=
  \begin{cases}
    t(e) & \text{if $t(e)\not\in \top$},\\
    e^\top & \text{if $t(e)\in \top$}.
  \end{cases}
\end{equation*}

That is, we duplicate initial and accepting states so that each has at most one outgoing or incoming transition, respectively,
except for states which are both initial and accepting which by (a) and (b) already have no transitions at all.

\begin{lemma}
  \label{l:gSTef}
  $\mcal{A}'$ is reduced and $L(\mcal{A})=L(\mcal{A}')$.
\end{lemma}

\begin{proof}
  We have $\bot'\cap\top'=\bot\cap \top$,
  thus there are no transitions that start or terminate at a state which is both initial and accepting.
  The other initial and accepting states are those of type $e_\bot$ and $e^\top$
  which satisfy (e) and (f) by construction.
  Thus $\mcal{A}'$ satisfies (a), (b), (e) and (f).

  If $\lambda(e)$ is a starter, then $t(e)\in\top$ (since $\mcal{A}$ satisfies (c)), so $t'(e)=e^\top\in\top'$.
  Thus $\mcal{A}'$ also satisfies (c).
  Similarly, $\mcal{A}'$ satisfies (d) and thus is reduced.
  Accepting paths in $\mcal{A}'$ correspond to those in $\mcal{A}$, so $L(\mcal{A})=L(\mcal{A}')$.
\end{proof}

By combining Lemmas \ref{l:PtoST} to \ref{l:gSTef} we obtain the following.

\begin{proposition}
  For every P-automaton there is a reduced gST-automaton that recognizes the same language.
\end{proposition}

Now let $\mcal{A}=(Q,E,s,t,\lambda,\mu,\bot,\top)$ be a reduced gST-au\-to\-ma\-ton.
We present two conversions of $\mcal{A}$ to a partial HDA.
The first is simpler and the second produces a strict pHDA.

Let $X(\mcal{A})$ be the partial HDA with cells $Q\sqcup E/\sim$, where $e\sim s(e)$ whenever $\lambda(e)$ is a terminator and $e\sim t(e)$ if $\lambda(e)$ is a starter. We put
$\ev(q)=\mu(q)$ for $q\in Q$, and for $e\in E$ with $\lambda(e)=\ilo{S}{U}{T}$, $\ev(e)=U$,
with only defined face maps
\begin{equation*}
  \delta^0_{U\setminus S}(e)=s(e),\quad
  \delta^1_{U\setminus T}(e)=t(e),
\end{equation*}
and $\bot_{X(\mcal{A})}=\bot_{\mcal{A}}$, $\top_{X(\mcal{A})}=\top_{\mcal{A}}$. (Note that if $e\sim q$, then $\ev(e)=\ev(q)$ and one of the equations above reduces to $\delta^\varepsilon_{\emptyset}(e)=q$.)

\begin{lemma}
  $X(\mcal{A})$ is a pHDA.
\end{lemma}

\begin{proof}
  We will check that no two face maps are composable, making condition \eqref{eq:lpcid} trivially satisfied. If $x\in X(\mcal{A})$ is a (non-trivial) lower face of $[y]$, then $x=[s(e)]$ and $y=[e]$ for $e\in E$ such that $\lambda(e)$ is not a terminator ($S_{\lambda(e)}\neq \lambda(e)$). If $s(e)\sim e'$ for some $e'$, then there are two possibilities:
  \begin{itemize}
  \item $\lambda(e')$ is a starter and $s(e')=s(e)=:q$. By condition (d), $q\in \bot$ and then $e'=e$ by condition (e). Since $\lambda(e)$ is not a terminator we get the contradiction.
  \item $\lambda(e')$ is a terminator and  $t(e')=s(e)=:q$. Then $q\in \top$ by condition (c) which contradicts (b).
  \end{itemize}
  Thus, $s(e)$ stands alone in its equivalence class  and hence has no faces. The case when $s(e)$ is an upper face is similar.
\end{proof}

\begin{lemma}
  $L(X(\mcal{A}))=L(\mcal{A})$.
\end{lemma}

\begin{proof}
  Any accepting path $(q_0,e_1,q_1,\dotsc,q_n)$ in $\mcal{A}$ translates to a path
  \begin{equation*}
    (q_0\arrO{A_1}e_1\arrI{B_1}q_1\arrO{A_2}\dotsm \arrI{B_n} q_n)
  \end{equation*}
  where $A_k=U\setminus S$, $B_k=U\setminus T$ for $\ev(e_k)=\ilo SUT$. (The first upstep becomes trivial if $\lambda(e_1)$ is a terminator; the last downstep, if $\lambda(e_n)$ is a terminator.)
  Conversely, by construction these are the only accepting paths in $X(\mcal{A})$.
\end{proof}

In order to convert $\mcal{A}$ to a \emph{strict} pHDA instead, we introduce a new variant of HDAs,
similar but incomparable to iHDAs.

\begin{definition}
  The \emph{cone precubical category} $\csq$ has objects $\isq$,
  and morphisms are generated by
  \begin{itemize}
  \item $d^0_A:\ilo{S}{(U\setminus A)}{U\setminus A}\to \ilo SUT$ for $\emptyset\neq A\subseteq U\setminus S$,
  \item $d^1_B:\ilo{U\setminus B}{(U\setminus B)}{T}\to \ilo SUT$ for $\emptyset\neq B\subseteq U\setminus T$
  \end{itemize}
  together with iconclist isomorphisms on each object.
  Composition is defined as in $\Box$.
\end{definition}

A \emph{cone HDA} is a presheaf $X: \csq^\op\to \Set$ together with start and accept cells $\bot, \top\subseteq X$.
The intuition is that $X$ has a set of special \emph{cone cells} $x$
which only have cone-shaped faces $\delta_A^0(x)$ and $\delta_B^1(x)$,
whereas all other cells form part of such cone.
As an example, Fig.~\ref{fig:cone_hda} shows a cone HDA with one cone cell $y_{ab, b}^*$.

\begin{lemma}
  Every HDA is a cone HDA, and every cone HDA is a strict partial HDA.
\end{lemma}

\begin{proof}
  The translations are the same as the ones from HDAs to iHDAs and from iHDAs to spHDAs in Section~\ref{se:rel}.
\end{proof}

To separate the classes, we provide several examples.
Figure~\ref{fig:cone_hda} shows a cone HDA which is not an HDA neither an iHDA. In an iHDA, when a corner does not exist, then at least one of the events of the square cannot be unstarted or terminated. Hence (at least) one of the bottom and right edges should not exist for this example to be an iHDA.
In Fig.~\ref{fig:diff_ihda_sphda} we have an spHDA which is not a cone HDA: in a cone HDA, when several events of a cell can be unstarted, then their union can be unstarted, which is not the case here as the lower left corner does not exist.
Finally, Fig.~\ref{fig:diff_hda_ihda} shows an iHDA which is not a cone HDA: in a cone HDA, a lower face cannot have an upper face, but here the lower right corner is an upper face of the lower edge.

Now let $\csq^{\ilo AUB}$ denote the presheaf represented by the cone cell $\ilo AUB$. 
Define the following cells of $\csq^{\ilo AUB}$:
\begin{itemize}
\item the middle cell $w^*_{\ilo AUB}=\id_{\ilo AUB}\in Y(\ilo AUB)[\ilo AUB]$,
\item for $A\neq U$, the initial cell $w^0_{\ilo AUB}=\delta^0_{U\setminus A}(w^*_{\ilo AUB})\in{}$
  $\smash{\csq^{\ilo AUB}[\ilo AAA]}$,
\item for $B\neq U$, the final cell $w^1_{\ilo AUB}=\delta^1_{U\setminus B}(w^*_{\ilo AUB})\in{}$
  $\smash{\csq^{\ilo AUB}[\ilo BBB]}$.
\end{itemize}
Note that the only cell of $\csq^{\ilo UUU}$ is the middle cell $w^*_{\ilo UUU}$.

\input{figures/cone_hda.tex}

Informally, the cone HDA $X_{\mcal{A}}$ corresponding to the reduced gST-automaton $\mcal{A}$ is obtained by replacing every transition $e$ of $\mcal{A}$ with a cone cell $\csq^{\lambda(e)}$ while preserving its states,  see Fig.~\ref{fig:cone_hda}.

In order to give a precise construction, define
\begin{itemize}
\item $E_0=\{e\in E\mid S_{\lambda(e)}=\lambda(e)\}$;
\item $E_1=\{e\in E\mid T_{\lambda(e)}=\lambda(e)\}$;
\item $E_*=E\setminus(E_0\cup E_1)$;
\item $Q_0=s(E_0)$, $Q_1=t(E_1)$, $Q_*=Q\setminus(Q_0\cup Q_1)$.
\end{itemize}
Note that $Q_0\subseteq \bot$ and $Q_1\subseteq \top$ by conditions (c) and (d) and $E_0\cap E_1=\emptyset$ because there are no silent transitions.

For $e\in E$ and $q\in Q_*$ denote $Y_e=\csq^{\lambda(e)}$ and $Z_q=\csq^{\id_{\mu(q)}}$.
Let $y_e^*$ be the middle cell of $Y_e$ and $z_q$ the unique cell of $Z_q$.
For $e\in E\setminus E_0$ let $y_e^0\in Y_e$ be the initial cell, and for $e\in E\setminus E_1$ let $y^1_e\in Y_e$ be the final cell. 
Define $X_{\mcal{A}}$ as the quotient
\begin{equation*}
    \coprod_{e\in E} Y_e\sqcup \coprod_{q\in Q_*} Z_q/\sim
\end{equation*}
by relations $y^0_e\sim z_{s(e)}$ for every $e\in E\setminus E_0$ and $y^1_e\sim z_{t(e)}$ for $e\in E\setminus E_1$.
The start and accept cells of $X_{\mcal{A}}$ are defined by 
\begin{gather*}
	\bot_{X_{\mcal{A}}}=\{z_{q}\mid q\in \bot_{\mcal{A}}\cap Q_*\}\cup \{y^*_e\mid e\in E_0\},\\
	\top_{X_{\mcal{A}}}=\{z_{q}\mid q\in \top_{\mcal{A}}\cap Q_*\}\cup \{y^*_e\mid e\in E_1\}.
\end{gather*}

Let $j_e:\csq^{\lambda(e)}=Y_e\to X$ denote the canonical injections. For $e\in E$ define a path
\[
	\beta_e
	=
	\begin{cases}
		(q^0_e \arrO{\lambda(e)\setminus S_{\lambda(e)}} q^*_e \arrI{\lambda(e)\setminus T_{\lambda(e)}} q^1_e) & \text{for $e\in E_*$}\\
		(q^*_e \arrI{\lambda(e)\setminus T_{\lambda(e)}} q^1_e) & \text{for $e\in E_0$}\\
		(q^0_e \arrO{\lambda(e)\setminus S_{\lambda(e)}} q^1_e) & \text{for $e\in E_1$}
	\end{cases}
\]

\begin{lemma}
\label{l:FacesInXA}
	Let $x$ be a cell of $X_{\mcal A}$.	
	If $x$ is both a lower and an upper face (of possibly two distinct cells) then $x=z_q$ for some $q\in Q$.
	If $x$ has both a lower and an upper face, then $x=y^*_e$ for $e\in E_*$. 
\end{lemma}
\begin{proof}
	In the first case, we have $x\in X_{\mcal A}[\ilo SUT]$ for $S\neq U$ and $T\neq U$. Every such cell has the form $y^*_e$ for $e\in E_*$.
	To prove the second statement, note that no cell in $Y_e$ (and in $Z_q$) is both a lower and an upper face of another cell. Thus, $x$ must be represented be a non-trivial abstraction class, and hence, by $z_q$ for $q\in Q_*$.
\end{proof}

\begin{lemma}
\label{l:PathsInXA}
	Every non-constant sparse accepting path $\alpha$ in $X_{\mcal A}$ has the form $\beta_{e_1}*\dotsc * \beta_{e_n}$ for $e_i\in E$, $t(e_i)=s(e_{i+1})$, $s(e_1)\in \bot$ and $t(e_n)\in\top$. If $\alpha$ is constant, then $\alpha=(z_q)$ for $q\in \bot\cap \top$.
\end{lemma}

\begin{proof}
	In light of Lemma \ref{l:FacesInXA}, 
	\[
		\alpha=\eta* (z_{q_0}\arrO{} y_{e_1} \arrI{} z_{q_1}\arrO{}\dotsm \arrO{} y_{e_k}\arrI{} z_{q_k})*\omega,
	\]
	where $\eta$ is either constant or an downstep and $\omega$ is constant or an upstep. Clearly, $s(e_i)=q_i$ and $t(e_i)=q_{i+1}$ and then $\alpha=\eta*\beta_{e_1}*\dotsm*\beta_{e_k}*\omega$. If $\eta$ is a downstep, then it must be $(y^e_*\arrI{} z_{t(e)})=\beta_e$ for $e\in E_0$: these are the only start cells that have upper faces. Similarly, $\omega$ is either constant or $\beta_e$ for $e\in E_1$.
\end{proof}

\begin{lemma}
  $L(X_{\mcal{A}})=L(\mcal{A})$.
\end{lemma}

\begin{proof}
  We have $\ev(\beta_e)=\lambda(e)$ for all $e\in E$. Thus Lemma \ref{l:PathsInXA} implies that $L(X_{\mcal{A}})\subseteq L(\mcal{A})$. To prove the converse, note that 
  if $\alpha=(q_0,e_1,q_1,e_2,\dotsc,q_n)$ is an accepting path in $\mcal{A}$, then $\beta=\beta_{e_1}*\dotsc*\beta_{e_n}$ is an accepting path in $X$ with
  $\ev(\beta)=\ev(\beta_{e_1})*\dotsm*(\beta_{e_n})=\ev(e_1)*\dotsm*(e_n)=\ev(\alpha)$.
\end{proof}

The construction of $X_{\mcal{A}}$ and the proof of its correctness does not require conditions (e)--(f): multiple transitions outcoming from a start cell become separated by the construction itself.

\section{A Kleene Theorem for Partial HDAs}

With the machinery developed in the previous section we are now able to give a simple proof for a Kleene-type theorem for (strict) pHDAs.
The proof proceeds along the lines of the standard Thompson algorithm, which should be a pleasant surprise
given the complexity of the proof of the Kleene theorem for HDAs (and iHDAs) in \cite{DBLP:journals/lmcs/FahrenbergJSZ24}.

The \emph{rational operations} on sets $L, M\subseteq \iPoms$ are
$L\cup M$,
$L M=\{ P*Q\mid P\in L, Q\in M, T_P\simeq S_Q\}$, and
$L^+=\bigcup_{n\ge 1} L^n$.
Compared to \cite{DBLP:journals/lmcs/FahrenbergJSZ24},
the definition of $L M$ does not include closure under subsumption,
given that pHDA languages are not subsumption-closed.

Let $\iST$ denote the set of starters and terminators.
The \emph{basic} ipomset languages are $\emptyset$ and $\{P\}$ for all $P\in \iST$.
Note that this is an infinite set.
The \emph{rational languages} are the smallest class of languages generated by the basic languages under the rational operations.

Again compared to \cite{DBLP:journals/lmcs/FahrenbergJSZ24},
note that we do not include parallel composition $\|$.
It comes at the price of expanding the set of basic ipomsets, going from singletons to any iconclist.

\begin{theorem}
  A set of ipomsets $L$ is rational iff there is a pHDA $X$ such that $L=L(X)$.
\end{theorem}

\begin{proof}
  The proof of the backwards direction is like in \cite{DBLP:journals/lmcs/FahrenbergJSZ24}:
  Convert $X$ to an ST-automaton $\mcal{A}$,
  forget the state labelling
  and add extra initial states to not recognize the empty word,
  then use the standard Kleene theorem to generate a rational (word) expression on $\iST$
  in which any subexpression $x^*$ may be converted to $\epsilon+x^+$
  and which, when concatenation is understood as gluing, generates $L$.

  For the forward direction, the basic languages are recognized by the empty iHDA
  and by the representable objects $\sq^P$ with start cell $Y(S_p)$ and accept cell $Y(T_P)$.
  For the rational operations, we may work with gST-automata.
  First, $L(\mcal{A})\cup L(\mcal{B})=L(\mcal{A}\sqcup \mcal{B})$ (the disjoint union).

  For concatenation, let $\mcal{A}$ and $\mcal{B}$ be reduced gST-automata
  and define $\mcal{A}*\mcal{B}=(Q, E, s, t, \lambda, \mu, \bot, \top)$ as follows:
  \begin{itemize}
  \item $Q=Q_\mcal{A}\cup Q_\mcal{B}$, $\mu=\mu_\mcal{A}\cup \mu_\mcal{B}$, $\bot=\bot_\mcal{A}$, $\top=\top_\mcal{B}$
  \item $E=E_\mcal{A}\cup E_\mcal{B}\cup \{(p^\top, q_\bot)\subseteq \top_\mcal{A}\times \bot_\mcal{B}\mid
    \mu_\mcal{A}(p^\top)\simeq \mu_\mcal{B}(q_\bot)\}$
  \item $s(e)=s_\mcal{A}(e)$ if $e\in E_\mcal{A}$, $s_\mcal{B}(e)$ if $e\in E_\mcal{B}$, and $s((p^\top, q_\bot))=p^\top$ otherwise,
    and similarly for $t$
  \item $\lambda(e)=\lambda_\mcal{A}(e)$ if $e\in E_\mcal{A}$, $\lambda_\mcal{B}(e)$ if $e\in E_\mcal{B}$,
    and $\lambda((p^\top, q_\bot))=\id_{\mu_\mcal{A}(p^\top)}$ otherwise
  \end{itemize}
  Then $L(\mcal{A}*\mcal{B})=L(\mcal{A}) L(\mcal{B})$ because of properties (a) and (b) of reduced gST-automata.

  Finally, for the Kleene plus, let $\mcal{A}$ be a reduced gST-automaton,
  then we define $\mcal{A}^+=(Q, E, s, t, \lambda, \mu, \bot, \top)$ as follows:
  \begin{itemize}
  \item $Q=Q_\mcal{A}$, $\mu=\mu_\mcal{A}$, $\bot=\bot_\mcal{A}$, $\top=\top_\mcal{A}$
  \item $E=E_\mcal{A}\cup \{(p^\top, q_\bot)\subseteq \top_\mcal{A}\times \bot_\mcal{A}\mid
    \mu_\mcal{A}(p^\top)\simeq \mu_\mcal{A}(q_\bot)\}$
  \item $s(e)=s_\mcal{A}(e)$ if $e\in E_\mcal{A}$ and $s((p^\top, q_\bot))=p^\top$ otherwise, similarly for $t$
  \item $\lambda(e)=\lambda_\mcal{A}(e)$ if $e\in E_\mcal{A}$ and $\lambda((p^\top, q_\bot))=\id_{\mu_\mcal{A}(p^\top)}$ otherwise
  \end{itemize}
  Then $L(\mcal{A}^+)=L(\mcal{A})^+$ because of properties (a) and (b).
\end{proof}

\section{Determinization}

In this section we show that every partial HDA can be determinized, that is, there exists  a deterministic pHDA that recognizes the same language. The following definition generalizes \cite[Def.\@ 6.1]{DBLP:journals/fuin/FahrenbergZ24}.

\begin{definition}
	A pHDA $X$ is \emph{deterministic} if
	\begin{enumerate}
	\item for every $U\in\sq$ there exists at most one initial cell in $X[U]$,
	\item for all $A\subseteq V\in \sq$ and $x\in X[V-A]$ there exists at most one cell $y\in X[V]$ such that $x=\delta^0_A(y)$.
	\end{enumerate}
\end{definition}

\input{figures/hda_not_determinizable}

\cite{DBLP:journals/fuin/FahrenbergZ24} shows that for HDAs, a language $L$ is accepted by a deterministic HDA iff it is \emph{swap invariant}: for all $P, Q, P', Q' \in \iPoms$ such that $PP' \in L, QQ' \in L$, and $P \sqsubseteq Q$, also $QP' \in L$.
Figure~\ref{fig:hda_not_determinizable} shows an HDA whose language $L$ is not swap invariant: $(ab)c \in L$, $\loset{a\\b} \in L$, and $ab \sqsubseteq \loset{a \\ b}$, but $\loset{a \\ b} c \not\in L$.
This HDA cannot be determinized as an (i/sp/sr)HDA.
Intuitively, in these four models, a square with a lower left and an upper right corner needs to contain all its four edges.
Then, there would be only one suitable path for $ab$ through these edges,
and extending it to obtain $abc$ would necessarily create a path that also recognizes $\loset{a \\ b} c$.

\input{figures/determinization_as_phda1}

\input{figures/determinization_as_phda2}

However, there are several ways to obtain deterministic pHDAs that recognize $L$.
We illustrate two of these in Figures~\ref{fig:determinization_as_phda1} and~\ref{fig:determinization_as_phda2}.
In the first one, we ``break the connection'' between an edge and its lower face, thus removing the only source of non-determinism.
In the second one, we consider that for each ipomset $P$ that is recognized, there is one path that recognizes $P$ that is more important than the others: the sparse path (see Lem.~\ref{le:sparse_existence}).
In the pHDA in Fig~\ref{fig:determinization_as_phda2}, the only \emph{existing} paths are sparse paths.
It is then deterministic, and recognizes $L$.
In the following, we show that such a deconstruction of a pHDA always exists, and that a power-set like construction can be applied to obtain a deterministic pHDA.

Let $\iPoms^q\subseteq \iPoms$ be the set of ipomsets $P$ such that the last element of its sparse decomposition $P=Q_1*\dotsm*Q_m$ is not a starter, that is, either $Q_m$ is a terminator or $m=0$ and $P$ is the identity ipomset.
Further, for $U\in \iPoms$ let $\iPoms^q_U=\{P\in \iPoms^q\mid T_P=U\}$.
The next lemma follows immediately from the uniqueness of sparse decompositions.

\begin{lemma}
  \label{l:IPomsetSplitting}
  Every ipomset $P$ has a unique presentation as $P=P'*\ilo VUU$, where $P'\in\iPoms^q$.
\end{lemma}

For a pHDA $X$, an ipomset $P$ and a subset $K\subseteq X$ define
\begin{multline*}
  X(K,P)=\{x\in X\mid \exists \alpha \text{ path in } X: \\
  \ev(\alpha)=P,\; \src(\alpha)\in K,\; \tgt(\alpha)=x\}
\end{multline*}
and let $\iPoms(X)=\{P\in\iPoms\mid X(\bot_X,P)\neq\emptyset\}$.

\begin{lemma}
  \label{l:TrackSetDecomp}
  For all $K\subseteq X$, $U\subsetneq V\in\sq$, $P\in \iPoms^q_U$ we have
  $X(K,P*\ilo UVW)=X(X(K, P), \ilo UVW)$.
\end{lemma}

\begin{proof}
  If $\alpha$ is a sparse path in $X$,
  then using Lemma \ref{le:decomp} there exists a decomposition $\alpha=\beta*\gamma$
  such that $\ev(\beta)=P$ and $\ev(\gamma)=\ilo UVW$.
  This shows that $X(K,P*\ilo UVW)\subseteq X(X(K, P), \ilo UVW)$. The inverse inclusion is obvious.
\end{proof}

\begin{lemma}
\label{l:pHDAWithProperStartCells}
For every pHDA $X$ there exists a language equivalent pHDA $X'$ such that for every $x\in \bot_{X'}$ and $A\subseteq \ev(x)$ the face $\delta^0_A(x)$ is undefined.
\end{lemma}
\begin{proof}
	Let $X'$ be the pHDA with set of cells
	\[X'=X\sqcup\{(x,C)\mid x\in X,\; C\subseteq \ev(x),\;  \delta^0_C(x)\in\bot_X\}.
	\]
	We extend the event by putting $\ev((x,C))=\ev(x)$,and face maps by $\delta^0_A((x,C))=(\delta^0_A(x),C-A)$ and $\delta_{A,B}((x,C))=\delta_{A,B}(x)$ whenever $B\neq\emptyset$; the faces $\delta^0_A((x,C))$ are undefined for $A\not\subseteq C$. Finally, we put $\bot_{X'}=\{(x,\emptyset)\mid x\in\bot_X\}$ and $\top_{X'}=\top_X\cup \{(x,A)\in X'\mid x\in\top_X\}$. Let $\alpha$ be a sparse accepting path in $X$. If $\alpha=(x_0\arrO{A}x_1\arrI{B}x_2\dotsm)$ starts with an upstep, then $((x_0,\emptyset)\arrO{A}(x_1,A)\arrI{B}x_2\dotsm)$ is an accepting path in $X'$ that recognizes $\ev(\alpha)$; if $\alpha=(x_0\arrI{B}x_1\dotsm)$ starts with a downstep, then $((x_0,\emptyset)\arrI{B}x_1\dotsm)$ is such a path. Thus, $L(X)\subseteq L(X')$. To prove the converse, note that the map $X'\to X$ that forgets the second coordinate ($x\mapsto x$, $(x,A)\mapsto x$) preserves accepting paths.
\end{proof}

Fix a pHDA $X$ that satisfies Lemma \ref{l:pHDAWithProperStartCells}.
The construction of a deterministic pHDA that recognizes $L(X)$ is an adaptation of the classical argument: its states are sets of cells that can be reached from a start cell by a path that recognizes a given ipomset. We take into account only ipomsets belonging to $\iPoms^q$.

Introduce an equivalence relation $\sim^X_U$ on $\iPoms^q_U(X)=\iPoms^q\cap \iPoms(X)\cap \iPoms_U$ by
\[
	P\sim^X_U Q \iff X(\bot_X, P)=X(\bot_X, Q).
\]
Let $[P]$ denote the equivalence class of $P$ in $\iPoms^q_U/{\sim^X_U}$.
Note that if $\id_U\in\iPoms^q_U(X)$, then $[\id_U]=\{\id_U\}$: we have $X(\bot_X,\id_U)=\bot_X\cap X[U]$, while Lemma \ref{l:pHDAWithProperStartCells} assures that no start cell belongs to $X(\bot_X,P)$ when $P$ is not an identity.

We define a pHDA $Y=\mathrm{Det}(X)$ as follows (face maps are defined only if both source and target is defined).
\begin{itemize}
\item
	$Y[U]$ is the set of triples $([P],A,U)$, where $A\subseteq U\in\sq$, $P\in \iPoms^q_{U\setminus A}(X)$, such that $P*\ilo {U\setminus A}UU\in\iPoms(X)$
\item
	$\delta^0_A([P],A,U)=([P],\emptyset,U\setminus A)$
\item 
	$\delta^1_B([P],A,U)=([P*\ilo{U\setminus A}{U}{U\setminus B}],\emptyset,U\setminus B)$ for $\emptyset\neq A,B\subseteq U$
\item
	$\delta^1_B([\id_V],\emptyset,V)=([\ilo{V}{V}{V\setminus B}],\emptyset,V\setminus B)$
\item
	no other face maps exist, except the compositions 
	\[
          ([\id_V],A,U) \smash{\xrightarrow{\delta^0_A}}
          ([\id_V],\emptyset,V) \smash{\xrightarrow{\delta^1_B}}
          ([\ilo{V}{V}{V\setminus B}],\emptyset,V\setminus B)
        \]
	for $\emptyset\neq A\subseteq U\in\sq$, $V=U\setminus A$
\item
	$\bot_Y=\{(\id_V,\emptyset,V)\mid \bot_X[V]\neq\emptyset\}$
\item
	$([P],A,U)\in\top_Y\iff P*\ilo{U\setminus A}{U}{U}\in L(X)$
\end{itemize}
Below we show that, for an ipomset $P\in\iPoms(X)$ with a presentation $P=P'*\ilo VUU$ from Lemma \ref{l:IPomsetSplitting}, there is exactly one path $\omega_P$ that starts at the start cell and recognizes $P$. Furthermore, $\tgt(\alpha_P)=([P'],U\setminus V,U)$.

\begin{lemma}
	The definition above is valid.
\end{lemma}
\begin{proof}
	We need to show that the upper face map is well defined, that is, if $P\sim_U^X Q$ for $P,Q\in\iPoms^q_U(X)$ and $\emptyset\neq A,B \subseteq U$, then $P*\ilo{U\setminus A}{U}{U\setminus B}\sim^X_U Q*\ilo{U\setminus A}{U}{U\setminus B}$. From Lemma~\ref{l:TrackSetDecomp} we have
\begin{multline*}
	X(\bot_X,P*\ilo{U\setminus A}{U}{U\setminus B})
	=
	X(X(\bot_X,P),\ilo{U\setminus A}{U}{U\setminus B})
	=\\
	X(X(\bot_X,Q),\ilo{U\setminus A}{U}{U\setminus B})	
	=
	X(\bot_X,Q*\ilo{U\setminus A}{U}{U\setminus B}).
\end{multline*}
We also need to prove that the final states are well-defined, namely that if
$P\sim^X_{U\setminus A}Q$ and $P*\ilo{U\setminus A}{U}{U}\in L(X)$ then $Q*\ilo{U\setminus A}{U}{U}\in L(X)$.
This is again a consequence of Lemma~\ref{l:TrackSetDecomp}:
\begin{multline*}
	X(\bot_X,Q*\ilo{U\setminus A}{U}{U})\cap\top_X
	=
	X(X(\bot_X,Q),\ilo{U\setminus A}{U}{U})\cap\top_X
	=\\
	X(X(\bot_X,P),\ilo{U\setminus A}{U}{U})\cap\top_X	
	=
	X(\bot_X,P*\ilo{U\setminus A}{U}{U})\cap\top_X\neq \emptyset.\qedhere
\end{multline*}
\end{proof}

\begin{lemma}
	$Y$ is deterministic.
\end{lemma}
\begin{proof}
	Immediately from the definition follows that the maps $\delta^0_A:Y[U]\to Y[U\setminus A]$ are injective and that $\lvert \bot_Y[U] \rvert \leq 1$ for all $U\in\sq$. 
\end{proof}

\begin{lemma}
\label{l:XInclY}
	Let $P\in\iPoms(X)$ be an ipomset with a decomposition $P=P'*\ilo VUU$.
	Then there exists a sparse path $\omega_P$ such that $\src(\omega_P)=(\id_{S_P},\emptyset, S_P)$, $\ev(\omega_P)=P$ and $\tgt(\omega_P)=([P'],U\setminus V,U)$.
\end{lemma}
\begin{proof}
	We define $\omega_P$ inductively with respect to the length of its sparse decomposition. If $P$ is an identity or a terminator we define
	\begin{itemize}
	\item for $P=\Id_V$ and we put $\omega_P=(\id_V,\emptyset,V)$,
	\item for $P=P'=\ilo WWU$ we put $\omega_P=((\id_W,\emptyset,W)\arrI{U\setminus W}(P,\emptyset,U))$,
	\end{itemize}
	If the sparse decomposition of $P$ contains a starter, we consider two cases:
	\begin{itemize}	
	\item If the sparse decomposition of $P$ terminates with a starter, then $P=P'*\ilo VUU$ for $V\subsetneq U$, and we put $\omega_P=\omega(P')*((P',\emptyset,V)\arrO{U\setminus V}(P',U\setminus V,U))$.
	\item If the last factor is a terminator, then $P=R*\ilo WVU$ for $W\subsetneq V\supsetneq U$ and 
  $R\in\iPoms^q$. We put 
  $\omega_P=\omega_{R*\ilo WVV} * (([R], V\setminus W,V)\arrI{V\setminus U}  ([P],\emptyset,U))$.
  \qedhere
	\end{itemize}
\end{proof}

\begin{lemma}
  \label{l:YInclX}
  If $\alpha$ is a path in $Y$ such that $\src(\alpha)\in\bot_Y$, then $\alpha=\omega_{\ev(\alpha)}$.
\end{lemma}

\begin{proof}
	Clearly $\alpha$ is sparse. Assume that $\alpha$ is a shortest path such that $\alpha\neq\omega_{\ev(\alpha)}$. 
  If $\alpha$ is constant, then $\alpha=([\id_V,\emptyset,V])=\omega_{\id_V}$. 
  Otherwise, there is a decomposition $\alpha=\beta*(x,\varphi,y)$. 
  Since $\beta$ is shorter than $\alpha$, we have $\beta=\omega_P$, $P=\ev{\beta}$. There are three cases to consider:
	\begin{itemize}
	\item $P=\id_V$ and $\beta=(([\id_V],\emptyset,V))$. Then either $\alpha=([\id_V],\emptyset,V)\arrO ([\id_V],A,U)=\omega_{\ilo VUU}$ for $A\subseteq U$, $V\cong U\setminus A$, or $\alpha=([\id_V],\emptyset,V)\arrI{B}([\ilo{U}{U}{U\setminus B}],\emptyset,U\setminus B)=\omega_{\ilo{U}{U}{U\setminus B}}$.
	\item $P\in\iPoms^q_U(X)$ is not an identity. Then $\tgt(\alpha)=([P],\emptyset,U)$. This cell has no defined faces. The only possibility is $\alpha=\omega_P*(([P],\emptyset,V)\arrO{A}([P],A,V))=\omega_{P*\ilo VUU}$ for $V=U\setminus A$.
	\item $P\not\in\iPoms^q_U(X)$ and is not an identity. Then $\tgt(\alpha)=([P],A,U)$, $A\neq0$. This cell is not a face of any other cell, so the only possibility is 
          $\alpha=\omega_P*(([P],A,U)\arrI{B}([P*\ilo{U\setminus A}{U}{U\setminus B}],\emptyset,U\setminus B))=\omega_{P*\ilo {U}{U}{U\setminus B}}$.
	\end{itemize}
	In all cases we yield a contradiction.
\end{proof}

\begin{proposition}
	$L(Y)= L(X)$.
\end{proposition}
\begin{proof}
	Note that $\tgt(\omega_P)\in\top_Y$ iff $P\in L(X)$. Thus the inclusion $ L(X)\subseteq  L(Y)$ follows from Lemma \ref{l:XInclY}, and the inverse from Lemma \ref{l:YInclX}.
\end{proof}

\section{Conclusion}

\begin{figure}
  \centering
  \newcommand*\sepa{\;---\;\ }
  \begin{tikzpicture}[x=1.8cm, y=1.2cm]
    \node (HDA) at (0,0) {HDA \sepa $\sq^\op\to \Set$};
    \node (iHDA) at (-1,-1) {iHDA \sepa $\isq^\op\to \Set$};
    \node (coHDA) at (1,-1) {cone HDA \sepa $\csq^\op\to \Set$};
    \node (spHDA) at (0,-2) {spHDA \sepa $\sq^\op\to \Set_*$};
    \node (srHDA) at (-1,-3) {\vphantom{p}srHDA \sepa $\sq^\op\to \Rel$};
    \node (pHDA) at (1,-3) {pHDA \sepa $\sq^\op\laxto \Set_*$};
    \node (rHDA) at (0,-4) {rHDA \sepa $\sq^\op\laxto \Rel$};
    \node (STA) at (0,-5) {STA}; 
    \node (gSTA) at (0,-6) {gSTA};
    \node (PA) at (0,-7) {PA};
    \path (HDA) edge (iHDA);
    \path (iHDA) edge (spHDA);
    \path (HDA) edge (coHDA);
    \path (coHDA) edge (spHDA);
    \path (spHDA) edge (srHDA) edge (pHDA);
    \path (srHDA) edge (rHDA);
    \path (pHDA) edge (rHDA);
    \path (rHDA) edge (STA);
    \path (STA) edge (gSTA);
    \path (gSTA) edge (PA);
  \end{tikzpicture}
  \caption{%
    Variants of HDAs, including new ones introduced in this paper.}
  \label{fig:hdaschema-new}
\end{figure}

We have brought order into the zoo of HDA variants, but in the process we have introduced several new ones.
Figure~\ref{fig:hdaschema-new} updates the spectrum from the beginning of the paper to include these new variants:
cone HDAs, generalized ST-automata, and P-automata.

We have shown that language-wise, the spectrum collapses into two classes,
HDAs and iHDAs on one side (languages closed under subsumption),
and all others, including the new cone HDAs,
whose languages are not necessarily closed under subsumption.
Indeed, the Kleene theorem of \cite{DBLP:journals/lmcs/FahrenbergJSZ24}
together with the one of the present paper
show that (i)HDA languages are generated by rational expressions with subsumption closure,
whereas the other class is generated by rational expressions \emph{without} subsumption closure.

We conjecture that also the Myhill-Nerode theorem \cite{DBLP:journals/fuin/FahrenbergZ24}
and the Büchi-Elgot-Trakhtenbrot theorem \cite{DBLP:conf/dlt/AmraneBFF24}
for HDAs have analogues for partial HDAs,
and due to the lack of subsumption closure,
we expect their proofs to be easier than the ones in \cite{DBLP:journals/fuin/FahrenbergZ24, DBLP:conf/dlt/AmraneBFF24}.

We have also shown that contrary to HDAs and iHDAs \cite{DBLP:journals/fuin/FahrenbergZ24},
partial HDAs may be determinized.
Further, the example in Fig.~\ref{fig:hda_not_determinizable} illustrates that neither spHDAs nor srHDAs are always determinizable.
It remains to be seen whether swap invariance is still a necessary and sufficient condition for determinizability.

It thus appears that from an automata-theoretic perspective,
partial HDAs are easier to work with than HDAs proper.
On the other hand, the \emph{geometry} of HDAs is much better behaved than the one of partial HDAs:
we have already seen that pHDAs may be ``geometrically broken'',
and the recent \cite{conf/fossacs/ChamounM26} shows that
geometric realizations of (lax) relational presheaves are difficult to understand from a geometric point of view.
To be more precise, \cite{conf/fossacs/ChamounM26} realizes relational presheaves as \emph{non-Hausdorff} topological spaces,
whereas realizations of HDAs are Hausdorff.
Whether other definitions are more suitable,
for example akin to \cite{DBLP:conf/fossacs/Dubut19} which breaks the directed structure rather than the topology,
is up to debate.

\bibliographystyle{plain}
\bibliography{mybib}

\end{document}

%% file: figures/phda_non_geometric.tex
\begin{figure}
\centering
\begin{tikzpicture}[x=1cm, y=1cm]
    \path[fill=black!10!white] (0,0) -- (2,0) -- (2,2) -- (0,2) -- (0,0);
    \node () at (1,1) {$x$};
    \node[state, rectangle] (00) at (0,0) {};
    \node[] (phantom1) at (1.7,2) {};
    \node[state, rectangle] (01) at (0,2) {};
    \node[state, rectangle] (10) at (2,0) {};
    \node[state, rectangle] (11) at (2,2) {};
    \path[] (00) edge node[below]  {$a$} (10);
    \path[draw] (01) -- (phantom1);
    \path[draw] (10) -- (11);
    \path[] (00) edge node[left]  {$b$} (01);
  \end{tikzpicture}
  \caption{pHDA such that $\delta_a^1 (\delta_b^1(x))$ is not defined.}
  \label{fig:phda_non_geometric}
\end{figure}

%% file: figures/broken_phda.tex
\begin{figure}
\begin{tikzpicture}[x=1.8cm, y=1cm]
	\path[use as bounding box] (0,-.2) -- (2,-.2) -- (2,2.6) -- (0,2.6) -- (0,-.2);
	\path[fill=black!10!white] (0,0) -- (2,0) -- (2,2) -- (0,2) -- (0,0);
	\node[state, rectangle, initial] (00) at (0,0) {};
	\node[state, rectangle] (10) at (2,0) {}; 
	\node[state, rectangle] (01) at (0,2) {}; 
	\node[state, rectangle] (11) at (2,2) {}; 
	\node[state, rectangle] (02) at (2,2.6) {}; 
	\node[coordinate] (01a) at (.6,2) {};
        \node at (11.east) {$\;\;\;\;\;\;m_1$};
        \node at (02.east) {$\;\;\;\;\;\;m_2$};
	\path (00) edge node[swap] {$a$} (10);
	\path (00) edge node {$b$} (01);
	\path (10) edge node[swap] {$b$} (11);
	\path (01) edge node {$a$} (02);
	\path (01a) edge node[swap, pos=.5] {$a$} (11);
	\node at (1,1) {$\loset{a\\b}\quad x$};
\end{tikzpicture}

\caption{%
  pHDA semantics of the Petri net in Fig.~\ref{fig:petritrans}.}
\label{fig:broken_phda}
\end{figure}

%% file: figures/diff_hda_ihda.tex
\begin{figure}
  \centering
  \begin{tikzpicture}[x=1cm, y=1cm]
    \path[fill=black!10!white] (0,0) -- (2,0) -- (2,2) -- (0,2) -- (0,0);
    \node () at (1,1) {$x$};
    \node[] (00) at (0,0) {};
    \node[] (01) at (0,2) {};
    \node[state, rectangle] (10) at (2,0) {};
    \node[state, rectangle] (11) at (2,2) {};
    \path[draw] (00) edge node[below]  {$a$} (10);
    \path[draw] (01) -- (11);
    \path[draw] (10) edge node[right]  {$b$}  (11);
  \end{tikzpicture}
  \caption{iHDA, but not HDA.}
  \label{fig:diff_hda_ihda}
\end{figure}

%% file: figures/diff_ihda_sphda.tex
\begin{figure}
\centering
\begin{tikzpicture}[x=1cm, y=1cm]
    \path[fill=black!10!white] (0,0) -- (2,0) -- (2,2) -- (0,2) -- (0,0);
    \node () at (1,1) {$x$};
    \node[] (00) at (0,0) {};
    \node[state, rectangle] (01) at (0,2) {};
    \node[state, rectangle] (10) at (2,0) {};
    \node[state, rectangle] (11) at (2,2) {};
    \path[draw] (00) edge node[below] {$a$} (10);
    \path[draw] (01) -- (11);
    \path[draw] (10) edge node[right] {$b$} (11);
    \path[draw] (00) -- (01);
  \end{tikzpicture}
  \caption{spHDA, but not iHDA.}
  \label{fig:diff_ihda_sphda}
\end{figure}

%% file: figures/diff_sphda_phda.tex
\begin{figure}
\centering
\begin{tikzpicture}[x=1cm, y=1cm]
    \path[fill=black!10!white] (0,0) -- (2,0) -- (2,2) -- (0,2) -- (0,0);
    \node () at (1,1) {$x$};
    \node[state, rectangle] (00) at (0,0) {};
    \node[state, rectangle] (01) at (0,2) {};
    \node[state, rectangle] (10) at (2,0) {};
    \node[state, rectangle] (11) at (2,2) {};
    \path[draw] (00) edge node[below]  {$a$} (10);
    \path[draw] (01) -- (11);
    \path[draw] (10) edge node[right]  {$b$} (11);
  \end{tikzpicture}
  \caption{pHDA, but not spHDA.}
  \label{fig:diff_sphda_phda}
\end{figure}

%% file: figures/diff_sphda_srhda.tex
\begin{figure}
\centering
\begin{tikzpicture}[x=1cm, y=1cm]
    \path[fill=black!10!white] (0,0) -- (2,0) -- (2,2) -- (0,2) -- (0,0);
    \node () at (1,1) {$x$};
    \node[state, rectangle] (00) at (0,0) {\tiny $u$};
    \node[state, rectangle, densely dashed] (00bis) at (.2,.2) {\tiny $v$};
    \node[] (phantom1) at (1.5,0) {};
    \node[] (phantom2) at (0,1.5) {};
    \node[state, rectangle] (01) at (0,2) {};
    \node[state, rectangle] (10) at (2,0) {};
    \node[state, rectangle] (11) at (2,2) {};
    \path[draw] (00) edge node[below]  {$a$} (10);
    \path[draw] (01) -- (11);
    \path[draw] (10) -- (11);
    \path[draw] (00) edge node[left]  {$b$} (01);
    \path[] (00bis.0) edge [-, dashed, bend right=10] (phantom1.west);
    \path[] (00bis.north) edge [-, dashed, bend left=10] (phantom2.south);
  \end{tikzpicture}
  \caption{srHDA, but not spHDA.}
  \label{fig:diff_sphda_srhda}
\end{figure}

%% file: figures/diff_srhda_rhda.tex
\begin{figure}
\centering
\begin{tikzpicture}[x=1cm, y=1cm]
    \path[fill=black!10!white] (0,0) -- (2,0) -- (2,2) -- (0,2) -- (0,0);
    \node () at (1,1) {$x$};
    \node[state, rectangle] (00) at (0,0.15) {\tiny $u$};
    \node[state, rectangle] (00bis) at (0,-0.15) {\tiny $v$};
    \node[] (00ter) at (0,0) {};
    \node[state, rectangle] (01) at (0,2) {};
    \node[state, rectangle] (10) at (2,0) {};
    \node[state, rectangle] (11) at (2,2) {};
    \path[] (00ter) edge node[below] {$a$} (10);
    \path[draw] (01) -- (11);
    \path[draw] (10) edge node[right] {$b$} (11);
  \end{tikzpicture}
  \caption{rHDA, but not srHDA.}
  \label{fig:diff_srhda_rhda}
\end{figure}

%% file: figures/diff_rhda_sta.tex
\begin{figure}
\centering
\begin{tikzpicture}[x=1cm, y=1cm]
    \node[state,circle] (00) at (0,0) {$\varepsilon$};
    \node[state,circle] (10) at (2,0) {$a$};
    \node[state,circle] (11) at (2,2) {$\loset{a \\ b}$};
    \path[draw] (00) edge node[below]  {$a\ibullet$} (10);
    \path[draw] (10) edge node[right]  {\tiny $\loset{\ibullet a\ibullet \\ b\ibullet}$} (11);
  \end{tikzpicture}
  \caption{ST-automaton, but not rHDA.}
  \label{fig:diff_rhda_sta}
\end{figure}

%% file: figures/cone_hda.tex
\begin{figure}
\centering
\begin{tikzpicture}[x=1cm, y=1cm]
    \path[fill=black!10!white] (0,0) -- (2,0) -- (2,2) -- (0,2) -- (0,0);
    \node () at (1,1) {$y^*_{ab,b}$};
    \node[state, rectangle]  (00) at (0,0) {};
    \node[] (01) at (0,2) {};
    \node[] (10) at (2,0) {};
    \node[] (11) at (2,2) {};
    \path[draw] (00) edge node[below]  {$a$} (10);
    \path[draw] (10) -- (11);
    \path[draw] (00) edge node[left]  {$b$} (01);
    \node[below left] at (00) {$y_\bot^{ab,b}$};
    \node[right] at (2,1) {$y^\top_{ab,b}$};
      \node[state,circle] (a00) at (-4,0) {$\varepsilon$};
      \node[state,circle] (a1h) at (-2,1) {$b$};
       \path[draw,bend left=20] (a00) edge node[above left]  {\tiny $\loset{\phantom{\ibullet} a \phantom{\ibullet} \\ \phantom{\ibullet} b \bullet}$ } (a1h);
  \end{tikzpicture}
  \caption{A gST-transition and its cone HDA
    $Y(\tiny \loset{a\pibullet \\ b \ibullet} )$.}
  \label{fig:cone_hda}
\end{figure}

%% file: figures/hda_not_determinizable.tex
\begin{figure}[bp]
\centering
\begin{tikzpicture}[x=1cm, y=1cm]
    \path[fill=black!10!white] (0,0) -- (2,0) -- (2,2) -- (0,2) -- (0,0);
    \node () at (1,1) {$x$};
    \node[state, rectangle] (00) at (0,0) {};
    \node[state, rectangle] (01) at (0,2) {};
    \node[state, rectangle] (10) at (2,0) {};
    \node[state, rectangle] (11) at (2,2) {};
    \node[state, rectangle] (20) at (4,0) {};
    \node[state, rectangle] (30) at (6,0) {};
    \path[] (00) edge node[below]  {$a$} (10);
    \path[] (10) edge node[below]  {$b$} (20);
    \path[] (20) edge node[below]  {$c$} (30);
    \path[draw] (01) -- (11);
    \path[draw] (10) -- (11);
    \path[] (00) edge node[left]  {$b$} (01);
    \node[] (bot) at (-0.3,0) {$\bot$};
    \node[] (top1) at (2.3,2) {$\top$};
    \node[] (top2) at (6,0.3) {$\top$};
  \end{tikzpicture}
  \caption{An HDA that cannot be determined as an srHDA.}
  \label{fig:hda_not_determinizable}
\end{figure}

%% file: figures/determinization_as_phda1.tex
\begin{figure}
\centering
\begin{tikzpicture}[x=1cm, y=1cm]
    \path[fill=black!10!white] (0,0) -- (2,0) -- (2,2) -- (0,2) -- (0,0);
    \node at (1,1) {$x$};
    \node[state, rectangle] (00) at (0,0) {};
    \node[] (phantom1) at (2,0.4) {};
    \node[state, rectangle] (01) at (0,2) {};
    \node[state, rectangle] (10) at (2,0) {};
    \node[state, rectangle] (11) at (2,2) {};
    \node[state, rectangle] (20) at (4,0) {};
    \node[state, rectangle] (30) at (6,0) {};
    \path[] (00) edge node[below]  {$a$} (10);
    \path[] (10) edge node[below]  {$b$} (20);
    \path[] (20) edge node[below]  {$c$} (30);
    \path[draw] (01) -- (11);
    \path[draw] (phantom1) -- (11);
    \path[] (00) edge node[left]  {$b$} (01);
    \node[] (bot) at (-0.3,0) {$\bot$};
    \node[] (top1) at (2.3,2) {$\top$};
    \node[] (top2) at (6,0.3) {$\top$};
    \node[] (top3) at (4,0.3) {$\top$};
  \end{tikzpicture}
  \caption{Determinization as a pHDA, ``breaking'' an edge.}
  \label{fig:determinization_as_phda1}
\end{figure}

%% file: figures/determinization_as_phda2.tex
\begin{figure}
\centering
\begin{tikzpicture}[x=.7cm, y=.7cm]
    \path[fill=black!10!white] (0,0) -- (1.4,-1.4) -- (0,-2.8) -- (-1.4,-1.4) -- (0,0);
    \node () at (0.5,-.9) {$\vphantom{b}a$};
    \node () at (-0.5,-.9) {$b$};
    \node[state, rectangle] (00) at (0,0) {};
    \node[state, rectangle] (10) at (2,0) {};
    \node[state, rectangle] (-10) at (-2,0) {};
    \node[state, rectangle] (-20) at (-4,0) {};
    \node[state, rectangle] (20) at (4,0) {};
    \node[state, rectangle] (30) at (6,0) {};
    \node[state, rectangle] (sqbot) at (0,-2.8) {};
    \path[] (00) edge node[below]  {$a$} (10);
    \path[] (10) edge node[below]  {$b$} (20);
    \path[] (20) edge node[below]  {$c$} (30);
    \path[] (00) edge node[below]  {$b$} (-10);
    \path[] (-10) edge node[below]  {$a$} (-20);
    \node[] (bot) at (-0,0.5) {$\bot$};
    \node[] (top2) at (6,0.5) {$\top$};
    \node[] (top3) at (4,0.5) {$\top$};
    \node[] (top4) at (-4,0.5) {$\top$};
    \node[] (top5) at (0,-3.3) {$\top$};
  \end{tikzpicture}
  \caption{Determinization as a pHDA, ``exploded'' to build an alternation of upsteps and downsteps.}
  \label{fig:determinization_as_phda2}
\end{figure}